\documentclass[a4paper]{easychair}

\usepackage{amssymb}
\newcommand{\solver}{RCP}
\usepackage{booktabs}   
\usepackage{subcaption} 
\usepackage{latexsym}
\usepackage{setspace}
\usepackage{cancel}
\usepackage{listings}
\usepackage{pgfplots}
\usepackage{graphicx}
\usepackage{appendix}
\usepackage{stmaryrd}
\usepackage{leftidx}
\usepackage{mathtools}
\usepackage{paralist}
\usepackage{color}
\usepackage{mathrsfs}
\usepackage{tikz}
\usetikzlibrary{shapes,automata,positioning,decorations.markings}
\usepackage[linesnumbered,ruled,noend]{algorithm2e}
\usepackage{wrapfig}
\usepackage{siunitx} 
\usepackage[ND,SEQ,EQ,ML]{prftree}
\usepackage{bussproofs}
\usepackage{xcolor,colortbl}

\definecolor{Lightblue}{rgb}{0.68,0.85,0.9}
\definecolor{Midblue}{rgb}{0.5,0.7,0.9}
\definecolor{Blue}{rgb}{0.4,0.6,0.9}
\definecolor{Gray}{gray}{0.85}
\definecolor{LightCyan}{rgb}{0.88,1,1}
\newcolumntype{a}{>{\columncolor{Gray}}c}
\newcolumntype{b}{>{\columncolor{white}}c}

\newtheorem{example}{Example}[section]
\newtheorem{definition}[example]{Definition}

\newtheorem{theorem}[example]{Theorem}
\newtheorem{lemma}[example]{Lemma}
\newtheorem{corollary}[example]{Corollary}

\newcommand{\OMIT}[1]{}

\newcommand{\ialphabet}{\Sigma}

\definecolor{light-gray}{gray}{0.9}
\definecolor{light-yellow}{RGB}{255, 255, 220}

\usepackage{pifont}
\usepackage[backgroundcolor=orange!50, textsize=scriptsize,textwidth=1cm]{todonotes}

\newcommand{\sidematt}[1]{{\color{orange} #1 -- MH}}
\definecolor{light-gray}{gray}{0.9}
\definecolor{light-yellow}{RGB}{255, 255, 220}

\newcommand{\akm}{\ensuremath{\left(a^m\right)^*}}

\lstdefinelanguage{JavaScript}{
	keywords={break, case, catch, const, continue, debugger, default, delete, do, else, export, finally, for, function, if, import, in, instanceof, let, new, return, super, switch, this, throw, try, typeof, var, void, while, with, yield},
	keywordstyle=\color{blue}\bfseries,
	ndkeywords={class, export, boolean, throw, implements, import, this},
	ndkeywordstyle=\color{gray}\bfseries,
	identifierstyle=\color{black},
	sensitive=false,
	comment=[l]{//},
	morecomment=[s]{/*}{*/},
	commentstyle=\color{gray}\ttfamily,
	stringstyle=\color{red}\ttfamily,
	morestring=[b]',
	morestring=[b]"
}

\newcommand\maineasychair[2]{#2}

\title{The Power of Regular Constraint Propagation (Technical Report)}

\author{
    Matthew Hague\inst{1}
    \and
    Artur Je\.z\inst{2}
    \and
    Anthony W.\ Lin\inst{3}
    \and
    Oliver Markgraf\inst{4}
    and
    Philipp R\"ummer\inst{5}
}

\institute{
    Royal Holloway University of London \\
    \email{matthew.hague@rhul.ac.uk}
    \and
    University of Wroclaw \\
    \email{aje@cs.uni.wroc.pl}
    \and
    University of Kaiserslautern-Landau and Max-Planck Institute \\
    \email{awlin@mpi-sws.org}
    \and
    University of Kaiserslautern-Landau \\
    \email{markgraf@cs.uni-kl.de}
    \and
    University of Regensburg and Uppsala University \\
    \email{philipp.ruemmer@ur.de}
}

\authorrunning{M. Hague, A. Jeż, A. W. Lin, O. Markgraf, P. R\"ummer}
\titlerunning{The Power of Regular Constraint Propagation}

\begin{document}

\maketitle

\begin{abstract}
    The past decade has witnessed substantial developments in string solving.
    Motivated by the complexity of string solving strategies adopted in existing
    string solvers, we investigate a simple and generic method 
    for solving string constraints: regular constraint propagation.
    The method repeatedly computes pre- or post-images of regular languages 
    under 
    the string functions present in a string formula, inferring more and more 
    knowledge about the possible values of string variables, until either a 
    conflict is found or satisfiability of the string formula can be 
    concluded. Such a propagation strategy is applicable to string constraints 
    with multiple operations like concatenation, replace, and almost all 
    flavors of string transductions. We demonstrate the generality and 
    effectiveness of 
    this method theoretically and experimentally.    
    On the theoretical side, we 
    show that RCP is sound and complete for a large fragment of string 
    constraints, subsuming both straight-line and chain-free constraints,
    two of the most expressive decidable fragments for which some modern string solvers 
    provide formal completeness guarantees.
    On the practical side, we implement regular constraint propagation within the open-source string solver OSTRICH. 
    Our experimental evaluation shows that this addition significantly improves OSTRICH’s performance and makes it competitive with existing solvers. In fact, it substantially outperforms other solvers on random PCP and bioinformatics benchmarks. The results also suggest that incorporating regular constraint propagation alongside other techniques could lead to substantial performance gains for existing solvers.
\end{abstract}

\section{Introduction} \label{sec:introduction}

Strings are a fundamental data type in many programming languages, heavily
utilized across popular languages like Python and JavaScript. These languages
are often equipped with rich string libraries, enabling a programmer to ``do
more with less code''. However, string manipulation is error-prone, leading to
security vulnerabilities such as cross-site scripting (XSS)
\cite{owasp13,owasp17,owasp21}.
Recent years have witnessed the emergence
of string solving as a promising automated method for reasoning about
string-manipulating programs, among others, the presence of security
vulnerabilities like XSS.
String solving has the theory of word equations as its foundation, a topic studied for over 50 years.
While many decidability results have been established (e.g., \cite{plandowski1999,jez2016,popl19,chain19}), their underlying proofs are often intricate and remain difficult to implement in practice.
The field also continues to face several long-standing open problems.

On the practical side of string solving, substantial progress has been made in recent years, with solvers adopting and optimizing foundational algorithms.
A prevalent method remains the strategy of \emph{splitting word equations} into multiple simpler equations, often combined with delayed handling of regular expressions~\cite{Makanin1977,Nielsen1917}, and implemented in modern tools such as Z3, Z3-alpha, Z3-Noodler, and OSTRICH~\cite{z3, z3alpha1,popl19,noodler-tool}.
This method, though effective, is far from straightforward: achieving good performance requires
carefully crafted heuristics that determine the most efficient way to split
equations and choose the branches to explore first, balancing the complexity and
efficiency of the solving process.
This is usually achieved by a combination
of other reasoning rules (e.g.\ involving length, regular languages, etc.).
Equation splitting is a fundamental method
used by most existing string solvers.

While equation splitting has been the dominant strategy so far, an alternative technique grounded in constraint propagation (CP)~\cite{rossi2006handbook} has also developed over time.
The first major application of this approach in string constraint solving can be traced back approximately 20 years ago \cite{cp03}, with early work on regular constraint propagation (RCP) techniques over fundamental string functions such as length, concatenation, prefix, suffix and regular constraints.
RCP, works by iteratively computing pre-images or post-images of regular languages
under the string functions present in a formula, propagating regular
constraints forwards and backwards. Propagation deduces increasingly more precise
information about the possible values of string variables,  until it either
encounters a conflict, indicating that the formula is unsatisfiable, or reaches
a point where the satisfiability of the string formula can be determined.

Over time, various solvers have incorporated RCP techniques to varying degrees, including STRANGER~\cite{stranger}, CertiStr~\cite{certistr}, OSTRICH~\cite{popl19,popl22}, and Z3-Noodler\cite{z3noodler,z3noodler1,noodler,noodler-tool}. However, these tools either lack completeness guarantees or support only restricted fragments.

Our work builds upon these prior efforts by refining and extending RCP techniques. We demonstrate that a subset of the proof system is sufficient to achieve completeness in the chain-free fragment, while also capturing more complex string functions such as \texttt{replaceAll}, transducers, reverse operations, and polyregular functions. Additionally, we provide an intuitive, algorithmic characterization of the chain-free fragment and analyze the limitations of the proof system, even when augmented with additional proof rules like equation splitting.

\maineasychair{\paragraph{Contributions.} 
We formalize RCP as a subset of the proof system from \cite{popl22}. 
%
%
Our \textbf{first contribution} is to show that RCP is complete for a large natural fragment of string constraints, which we call the \emph{orderable fragment}.
This fragment subsumes two of the largest known decidable
fragments called the \emph{chain-free string constraint fragment}
\cite{chain19} (over concatenation and rational transducer functions) and the 
\emph{straight-line fragment} \cite{popl19} (over a more general class of string functions satisfying
some semantic conditions). In doing so, we discover a substantially simplified
formulation of the chain-free fragment in \cite{chain19}.
Our \textbf{second contribution} is the implementation of RCP within the 
state-of-the-art solver OSTRICH \cite{ostrichsmtcomp}. We demonstrate a significant improvement in OSTRICH’s performance, and our experiments indicate that our RCP-based solver is the only string solver capable of solving random PCP benchmarks. 
Moreover, when combined with other solvers, our approach reduces the number of unsolved benchmarks by improvements ranging from 25\% up to 74\%, thereby complementing existing techniques.
Finally, given the power of regular constraint propagation, it is natural
to wonder if it can prove unsatisfiability for all unsatisfiable string
constraints with only word equations and regular constraints, which would give
rise to a simple alternative for complete algorithms (e.g.\ Makanin's algorithm
\cite{Makanin1977} and recompression algorithm \cite{jez2016}). 
Our \textbf{third contribution} is a result exhibiting this theoretical
limitation of RCP, even together with some other commonly used proof rules. 
In particular, we 
enrich RCP with two simple proof rules: (1) Nielsen's transformation
rule (which can be used to show decidability quadratic word equations
\cite{diekert}), and (2) The ``Cut'' rule \cite{popl22}. This essentially amounts to the
proof system of \cite{popl22} enriched with Nielsen's transformation.
We show that there are string
constraints with concatenation and regular constraints that are unsatisfiable
but cannot be proved using only RCP, Nielsen, and Cut. To the best of our 
knowledge, this constitutes one of the handful of impossibility results for a 
string solving strategy, especially with a rich set of proof rules.

}{}

\paragraph{Organization.} The paper is structured as follows. Section~\ref{sec:motivating-example} gives three motivating examples.
Section~\ref{sec:proof_sys} formalizes regular constraint propagation as a proof system.
Section~\ref{sec:guarantee} discusses completeness and decidability.
Theoretical limitations are addressed in Section~\ref{sec:lower}.
Finally, Section~\ref{sec:expr} presents experimental results, Section~\ref{sec:related-work} discusses related work and Section~\ref{sec:conc} concludes with remarks on future work.

\section{Motivating Examples}\label{sec:motivating-example}

\subsection{String Sanitization}

To motivate the challenges in reasoning about real-world string-processing programs, we consider a transformation that normalizes decimal strings by first applying sanitization and then rewriting and restructuring the content.

Consider the function shown in Figure~\ref{fig:normalize-decimal}, which takes a string representation of a decimal number, removes leading and trailing whitespace, eliminates redundant zeros, and ensures a canonical form. The function would, for instance, normalize \texttt{"   000123.45000   "} to \texttt{"123.45"}.
This function can be fully expressed as a string constraint and solved within a constraint-solving framework. 
By adding an assertion that the output must be in normalized form, a solver can determine whether all inputs correctly transform into their canonical representation. 
This enables automated reasoning about the correctness and behavior of the function, allowing us to verify that any input will always yield a properly formatted decimal number.

\begin{figure}[tb]
	\begin{lstlisting}[language=JavaScript]
		function normalize(decimal) {
			decimal = decimal.trim();
			const decimalReg = /^(\d+)\.?( \d*)$/;
			var decomp = decimal.match(decimalReg);
			var result = "";
			if (decomp) {
				var integer = decomp[1].replace(/^0+/, "");
				var fractional = decomp[2].replace(/0+$/, "");
				if (integer !== "") result = integer; else result = "0";
				if (fractional !== "") result = result + "." + fractional;
			}
			return result;
		}
	\end{lstlisting}

        \vspace{-4ex}
	\caption{Normalize a decimal by trimming whitespace and removing leading and trailing zeros.}
	\label{fig:normalize-decimal}
\end{figure}

It should be noted that the string transformations in
Figure~\ref{fig:normalize-decimal} are challenging for existing SMT
solvers.  With the most widely used SMT solvers, Z3~\cite{z3} andf cvc5~\cite{cvc5}, the
\lstinline!decimal.match! step could only be modeled using the
operations~\texttt{str.indexof} and \texttt{str.substr}.  Z3, cvc5 and Z3-noodler
are not complete for string constraints involving those operations,
however, so that no guarantees can be given that the solvers are able
to verify desired properties of the \lstinline!normalize! function.
Related work implemented in the SMT solver OSTRICH~\cite{popl22} supports similar functionality using prioritized streaming transducers. These transducers can express the \lstinline!decimal.match! operation, but result in a more complex and comparatively slower solving procedure.
We show that RCP gives rise to a much simpler decision procedure that is still able to precisely model the transformations in Figure~\ref{fig:normalize-decimal}.

\OMIT{
Previous work~\cite{popl22} studied a similar function but did not include an initial sanitization step. 
direct regular expression matching and replacement but did not account for cases where user input contained extraneous whitespace. This omission makes the function less robust for real-world applications such as user input validation, database normalization, and security-sensitive string comparisons. Including a sanitization step significantly expands the range of possible inputs and introduces additional complexity in reasoning about the transformation.
}

The function can be expressed in terms of functional
transformations. The first step applies a trimming function
$f_{\text{trim}}$, which removes leading and trailing whitespace. The
result is then decomposed into an integer part $y$ and a fractional
part $ z$, separated by a decimal point. The transformation proceeds
by applying two additional functions: $g_{\text{lead}}$ to remove
leading zeros from $y$ and $ g_{\text{trail}}$ to remove trailing
zeros from $ z $.

This transformation can be expressed using the following constraints:

\[
\begin{aligned}
	& \text{decimal} \in \text{decimalReg} \land\mbox{} \\
	& \text{decimal} = f_{\text{trim}}(\text{input}) \land\mbox{} \\
	& \text{decimal} = \text{integer} \mathbin{++}  \texttt{"."} \mathbin{++} \text{fractional} \;\land\mbox{} \\
	& \text{result} = g_{\text{lead}}\text{(integer)} \mathbin{++}  \texttt{"."} \mathbin{++} g_{\text{trail}}\text{(fractional)} \land\mbox{} \\
        & \text{result} \not\in \text{correctFormat}
\end{aligned}
\]
where \( f_{\text{trim}} \), \( g_{\text{lead}} \), and
\( g_{\text{trail}} \) can all be expressed as functional transducers.
In this formulation, we use $++$ to denote string concatenation, and
assert that the result string does not match some regular expression
``$\text{correctFormat}$'' to verify that no execution exists in which
string normalization fails.

The constructed formula is unsatisfiable, but beyond the fragments
decided by most of today's string solvers. Solvers like Z3, cvc5, or
Z3-noodler are not able to handle
\( f_{\text{trim}} \), \( g_{\text{lead}} \), and
\( g_{\text{trail}} \) and again have to resort to functions like
\texttt{str.indexof} and \texttt{str.substr} to encode those
transformations. The solver OSTRICH~\cite{popl19,popl22} supports
reasoning over transducers and \texttt{replace}/\texttt{replaceAll},
but requires input formulas in the \emph{straight-line} fragment to
guarantee completeness. Our formula is not straight-line, since two
equations with ``decimal'' as the left-hand side exist, causing
OSTRICH to fail in solving the constraint. We will
show that the formula is in a new fragment proposed in this paper,
denoted the \emph{orderable} fragment, which generalizes both the
straight-line and chain-free fragments. Since RCP is complete for
orderable formulas, it gives rise to a decision procedure for formulas
like the one shown here and can easily show the formula to be
unsatisfiable.
We revisit this example in Section \ref{sec:benchmarks}, Table \ref{tab:motivating-examples}, where \solver{} is the only solver to handle it successfully.
\OMIT{
Other work, such as OSTRICH~\cite{popl19,popl22}, supports reasoning over transducers, \texttt{replaceAll}. However, the inclusion of a preprocessing step such as $ f_{\text{trim}} $ introduces an additional assignment to the variable \texttt{decimal}, which modifies the structure of the constraint. Specifically, this transformation pushes the formula outside the straight-line fragment, causing OSTRICH to fail in solving the constraint. 
Most existing string solvers are not equipped to encode the specific functions required by this example—such as trimming whitespace or selectively removing leading and trailing zeros—in a precise and modular way.
While certain operations may be expressible in specialized or approximate forms in earlier solvers, they typically fall outside the scope of modern SMT-LIB-based solver support.
We demonstrate that, despite moving beyond the straight-line fragment, this formula remains decidable within our proposed fragment, which generalizes both the straight-line and chain-free fragments.
}

\maineasychair{\subsection{Post's Correspondence Problem}
\label{ex:pcp}
\emph{Post's Correspondence Problem (PCP)} is a well-known \textit{undecidable problem} that asks whether a given set of \textit{domino-like word pairs} (tiles) can be arranged in sequence such that the concatenation of the top and bottom sequences produces the same string.

In general, PCP instances consist of \textit{multiple dominos}, and the word lengths on either side of a domino may differ. However, even for highly restricted cases, the problem remains undecidable.

To illustrate, consider the following \textit{simplified instance}, consisting of a \textit{single} domino: {$\frac{10}{01}$}.
This means we are given the tile  {$\frac{10}{01}$} and need to determine whether any non-empty sequence of these tiles produces the same string on both top and bottom.

In string constraint solving, we encode this problem using a \textit{string variable} $x$ that selects repeated instances of this domino. Specifically:
\begin{equation*}
	\begin{aligned}
		&\text{x} \in \texttt{"2"}^+ \land 
		\text{y} = \texttt{replaceAll}(\text{x}, \texttt{"2"}, \texttt{"10"}) \; \land \\
		&\text{z} = \texttt{replaceAll}(\text{x}, \texttt{"2"}, \texttt{"01"}) \land \text{y} = \text{z}
	\end{aligned}
\end{equation*}
Here, $x$ is a \textit{non-empty string} over the symbol \texttt{"2"} (denoted \( \texttt{"2"}^+ \)), which represents repeated selections of the given domino. The \texttt{replaceAll} function \textit{maps} each occurrence of \texttt{"2"} to its corresponding \textit{top (\texttt{"10"}) and bottom (\texttt{"01"})} strings.
This formulation enforces that applying the same sequence of replacements to both top and bottom strings must yield the same result.
However, this instance is unsatisfiable—no sequence of replacements can equate the top and bottom string.
By propagating the regular constraints from $x$ forward onto $y$ and $z$, and using the fact that $y = z$, we find that any solution for $y$ and $z$ must lie in $(\texttt{"10"})^+ \cap (\texttt{"01"})^+$. Since this intersection is empty, the formula is unsatisfiable.
Notably, despite its simplicity, this small PCP instance is already beyond the reach of state-of-the-art string solvers. See Table~\ref{tab:motivating-examples} in Section~\ref{sec:benchmarks} for solver performance on this example.

\subsection{Reverse Transcription in Bioinformatics}

This example models a reverse transcription process inspired by bioinformatics~\cite{mount2004bioinformatics,compeau2015bioinformatics}. Here, an unknown RNA string $y$ is converted into a DNA string by applying a series of \texttt{replaceAll} operations that simulate nucleotide base pairing. In addition, the RNA string is required to contain a specific pattern.

The following formula captures an instance of this problem (where the \texttt{...}, as in \texttt{"TGAGTAT..."} and \texttt{"ucuc..."}, indicate longer concrete strings omitted for brevity):
\begin{equation*}
	\begin{aligned}
		&\text{x} = \texttt{"TGAGTAT..."} \land \\
		&\text{y1} = \texttt{replaceAll}(\text{y}, \texttt{"u"}, \texttt{"A"}) \; \land \\
		&\text{y2} = \texttt{replaceAll}(\text{y1}, \texttt{"a"}, \texttt{"T"}) \; \land \\
		&\text{y3} = \texttt{replaceAll}(\text{y2}, \texttt{"g"}, \texttt{"C"}) \; \land \\
		&\text{x} = \texttt{replaceAll}(\text{y3}, \texttt{"c"}, \texttt{"G"}) \; \land \\
		&\text{z} = \texttt{"ucuc..."} \; \land \\
		&\texttt{str.contains(y,z)}
	\end{aligned}
\end{equation*}

In this instance, the solver must determine an RNA string $y$ such that, after sequentially replacing \texttt{"u"} with \texttt{"A"}, \texttt{"a"} with \texttt{"T"}, \texttt{"g"} with \texttt{"C"}, and \texttt{"c"} with \texttt{"G"}, the resulting DNA string $x$ matches a given constant. Moreover, $y$ must contain the RNA pattern specified by $z$, as enforced by \texttt{str.contains}(y, z). This example illustrates another application of complex string transformations in a biologically motivated context. Importantly, this instance falls within the straight-line fragment and can be solved by OSTRICH as well as by our solver. However, it also highlights the need for robust support for complex string functions such as \texttt{replaceAll}—a feature that receives only limited support in other state-of-the-art solvers. 
Solver performance on one instance of this benchmark is summarized in Table~\ref{tab:motivating-examples} in Section~\ref{sec:benchmarks}.

}{}


\section{Regular Constraint Propagation as a Proof System}
\label{sec:proof_sys}

\subsection{String Constraint Fragment}

In this work, we focus on a subset of string constraint language that is
most relevant for regular constraint propagation. In particular, we do not
deal with length constraints, which can already be effectively dealt with
\emph{after} regular constraint propagation is done, for instance, by
looking at the length abstraction of the equational constraints, regular
constraints, together with the length constraints (e.g.\ see
\cite{z3noodler,ostrichsmtcomp,popl16}).
%
For simplicity, we also assume that formulas are provided in a \emph{normal
form}. The syntax of this normal form is defined as follows:
\begin{equation*}
	\psi ~::=~ \phi \mid \psi \land \psi \qquad\qquad
	\phi ~::=~ x \in e \mid x = f(x_1, \dots, x_n)
\end{equation*}
This normal form is not restrictive, as any formula in the supported string
language can be systematically transformed into an equivalent formula in this normal form.

To clarify the components of this normal form, we denote string variables by $x, x_1, \ldots, x_n$ (and later using $y$ and $z$).
A string constraint $\psi$ is a conjunction of string formulas~$\phi$.
There are two types of string formulas: \emph{regular constraints} and \emph{equational constraints}.

\emph{Regular constraints} are regular membership tests~$x \in e$ that assert that the value assigned to $x$ must belong to the regular language defined by a regular expression (RegEx)~$e$.
The exact supported regular expression syntax is not expanded on here, as it does not affect decidability.
However, it may influence the applicability or efficiency of certain proof rules discussed later, depending on whether features like intersection or complement are supported.
In an implementation, one could use equivalent representations of regular
languages including (nondeterministic) finite automata.

\emph{Equational constraints} are formulas~$x = f(x_1, \ldots, x_n)$ that assert that the value of $x$ is equal to the result of applying string function $f$ to the values of $x_1, \ldots, x_n$.
We admit general string functions $f: (\ialphabet^*)^n \to (\ialphabet^*)$ (in
fact, even when $f$ is a relation), but our approach works especially under the
assumption of at least one of \emph{forwardability} and \emph{backwardability},
which will be expounded below. A plethora of string functions satisfy at
least one of these conditions including concatenation, \texttt{replace},
\texttt{replaceAll}, \texttt{reverse}, and functions defined by various flavors of transducers (e.g.~rational,
two-way, and streaming transducers) (e.g.~\cite{popl19,z3noodler,berstel,CCHLW18,AC10,AD11}).


\subsection{The Proof System}
We introduce the \emph{RCP proof system}, which is based on propagating regular languages. The system operates on \emph{one-sided sequents} (i.e.\ left-sided) and can be interpreted similarly to \emph{Gentzen-style sequents} \cite{Gentzen1935}.
A \emph{sequent} in this system is a string constraint, denoted as $ \Gamma = \{ \phi_1, \phi_2, \dots, \phi_n \}$, where each \( \phi_i \) is a string formula;
it should be understood as a conjunction $\phi_1\land \phi_2 \land  \dots \land \phi_n$.

For convenience, we sometimes write sequents as a list $ \Gamma, \phi_1, \dots, \phi_n $ which represents the union $\Gamma \cup \{ \phi_1, \dots, \phi_n \}$,
and when writing a string constraint, we sometimes write $\{\phi_1, \phi_2\}$ instead of $\phi_1 \land \phi_2$.
A string constraint $\phi_1 \land \cdots \land \phi_n$ naturally corresponds to a sequent $\{ \phi_1, \ldots, \phi_n \}$.

A \emph{proof} (or \emph{proof tree}) is a \emph{rooted tree} $(V, E)$ where the nodes $V$ represent sequents.
For the sake of presentation, we omit an explicit labeling function from nodes to sequents and assume that each node is uniquely identified within the tree. In particular, if the same sequent appears more than once in the proof tree, we treat each occurrence as a distinct node.

The root of the tree is located at the \emph{bottom} and corresponds to the sequent whose unsatisfiability we want to prove,
while the \emph{leaves} at the \emph{top} of the tree are called \emph{axioms}, and represent unsatisfiable instances.
We say a proof tree is \emph{closed} if and only if all leaves correspond to the empty sequent $\Gamma = \emptyset$.
The edge relation of the proof tree is given by the proof rules.
Except for the initial sequent, every sequent in the proof tree must be \emph{derived} by one of the proof rules shown in Figure~\ref{fig:proof-rules}.

A \emph{proof rule} in our system takes the form  $\frac{S_1  S_2 \dots  S_n}{S}$,
indicating that the sequent $S$ may be proved from a sequence of sequents $S_1, S_2, \dots, S_n$.
If a sequent $S$ is derived from sequents $S_1, \dots, S_n$ by a proof rule, then $S$ is the parent node of $ S_1, \dots, S_n$ in the proof tree.
Formally, the edge relation $E \subseteq V \times V$ is defined so that for every proof step, we have  $(S, S_1), \dots, (S, S_n) \in E.$

Our \emph{RCP proof system} consists of exactly the proof rules, as given in
Figure \ref{fig:proof-rules};
note that $\times$ in \texttt{[Bwd-Prop]} represents the Cartesian product.
Before explaining the proof rules, we note that the proof system is a fragment of the proof system introduced in~\cite{popl22}, from which we obtain that it is \emph{sound}.

\begin{lemma}[Soundness~\cite{popl22}]
	\label{lem:RCP_is_sound}
    The RCP proof system is sound.
    That is, the root of a~closed proof is an unsatisfiable sequent.
\end{lemma}

\begin{figure}[tb]
\texttt{[Close]}
\[
    \prftree[r]{
        if $L(e_1) \cap \dots \cap L(e_n) = \emptyset$
    }{
        \phantom{all done}
    }{
        \Gamma, x \in e_1,\dots, x \in e_n
    }
\]

\texttt{[Intersect]}
\[
\prftree[r]{
	if $L(e) = \bigcap_{i=1}^n L(e_i)$
}{
	\Gamma, x \in e
}{
	\Gamma, x \in e_1, \dots, x \in e_n
}
\]
\texttt{[Fwd-Prop]}
\[
    \prftree[r]{
        if $L(e) = f(L(e_1),\dots,L(e_n))$
    }{
        \Gamma, x \in e, x = f(x_1,\dots, x_n), x_1 \in e_1,\dots, x_n \in e_n
    }{
        \Gamma, x_1 \in e_1, \dots, x_n \in e_n, x = f(x_1,\dots, x_n)
    }
\]
\texttt{[Bwd-Prop]}
\[
    \prftree[r]{
        $\begin{array}{l}
            \text{if } f^{-1}(L(e)) \\
            \ = \bigcup^k_{i=1} (L(e^i_1)\times \dots \times L(e^i_n))
        \end{array}$
    }{
        \left\{
            \begin{array}{c}
                \Gamma,
                x \in e,
                x = f(x_1,\dots, x_n), \\
                x_1 \in e^i_1,\dots, x_n \in e^i_n
            \end{array}
        \right\}^k_{i=1}
    }{
        \Gamma, x\in e, x = f(x_1,\dots, x_n)
    }
\]
\caption{\label{fig:proof-rules}Key proof rules of RCP}
\end{figure}

We explain the proof rules.
The $\texttt{Close}$ rule terminates a branch when a contradiction is found.
That is, a variable is subject to an unsatisfiable combination of regular constraints.
The $\texttt{Intersect}$ rule computes the intersection of regular expressions.
Forwards and backwards constraint propagation derive new regular constraints on string variables that arise due to function application.
Forwards propagation transmits regular constraints from function inputs to function outputs, while backwards propagation propagates constraints from function outputs back to their inputs.
Details and examples are given below.
 
For example, when \( x = \texttt{concat}(y, z) \), we can combine any known regular constraints \( y \in e_1 \) and \( z \in e_2 \) to derive that \( x \in e_1 e_2 \).  
For readability, we will from now on write \( x = yz \) instead of \( x = \texttt{concat}(y, z) \) to denote string concatenation.
For this forward propagation, we require that the image of a function on
regular languages is also regular, which is the case when concatenating two different variables.
In general, the forward-propagated constraint is only an \emph{overapproximation} of the true constraint on $x$:
Consider $x=f( y ) \wedge y \in e$, where $f( y ) = yy$,
in this case, the obtained regular constraint $x \in e e$ is a superset 
of the true constraint, which is $\{ ww : w \in L(e)\}$.
Although overapproximations can be easily added in our RCP
implementation (and are in fact sound when proving unsatisfiability), we do not
admit this in our theoretical study and hence not in our proof system above.

In the backwards direction, when $x =yz$, a regular constraint $x
\in e$ can be split into a finite disjunction of conjuncts of the form
$y \in e_1^i$ and $y \in e_2^i$ for any $\{e_1^i\}_{i=1}^k$ and $\{e_2^i\}_{i=1}^k$
such that $L(e) = \bigcup_{i=1}^k L(e_1^i e_2^i)$.
For example, if $x \in a^*b^*$, one potential disjunction could be
\[
    (y \in a^* \wedge z \in a^*b^*) \vee
    (y \in a^*b^* \wedge z \in b^*).
\]
The same in fact is true for $x=f( y ) \wedge
y \in e$, with  $f( y ) = yy$.
The condition for a finite number of alternatives for $e_1$ and $e_2$
is expressed by the side-condition
$f^{-1}(L(e))  = \bigcup^k_{i=1} (L(e^i_1)\times \dots \times L(e^i_n))$
and causes branching in the proof trees (all $k$ alternatives are considered).
\OMIT{
Intuitively, the ``splitting point'' of the constraint $e$ on $x$ can appear in multiple places.
For example, if $e = ab$, the split may occur before the $a$, between the $a$ and $b$, or after the $b$.
}

As mentioned, propagations can be performed for functions that allow these
in at least one of the two directions. We formalize this in terms of
forwardability and backwardability.
We lift string functions $f: (\Sigma^*)^k \to \Sigma^*$ to act on languages in the standard way.


\begin{definition}[Forwardable]
	A function $f: (\Sigma^*)^k \to \Sigma^*$ is \emph{forwardable} if, for all regular languages $L_1, \ldots, L_k \subseteq \Sigma^*$, the image $f(L_1, \ldots, L_k)$ is also a regular language, and there exists an algorithm that, given representations of $L_1, \ldots, L_k$, computes a representation of $f(L_1, \ldots, L_k)$.
\end{definition}

\begin{definition}[Backwardable]
	A function $f: (\Sigma^*)^k \to \Sigma^*$ is \emph{backwardable} if, for all regular languages $L \subseteq \Sigma^*$, the preimage $f^{-1}(L)$ is a recognizable relation~\cite{berstel}, and there exists an algorithm that, given a representation of $L$, computes a representation of $f^{-1}(L)$.
\end{definition}

The recognizable relation in the definition corresponds exactly to the notion in \texttt{Bwd-Prop}:
$R \subseteq (\Sigma^*)^n$ is recognizable, if it can be represented as
$R = \bigcup_{i=1}^k L_{i,1} \times \cdots \times L_{i,k}$,
with each $L_{i,j}$ being a regular language.
Note that such a representation is not unique.

The following lemmas summarize conditions on concatenation functions to be
forwardable.
\begin{lemma}
    Suppose $t = f(x_1,\ldots,x_n)$ is a term over the function symbol
    $\textnormal{\texttt{concat}}$ (possibly with string constants). Then, if each $x_i$
    appears at most once in $t$, then the string function associated with $t$
    is forwardable.
\end{lemma}
The proof is a simple application of the closure of regular languages under
concatenation. In the case when a variable appears multiple times in $t$, we
have also remarked above that the associated function is generally not
forwardable. That said, string functions using unrestricted concatenation and string constants are always
backwardable.
\begin{lemma}
    Suppose $t = f(x_1,\ldots,x_n)$ is a term over the function symbol
    \textnormal{\texttt{concat}} (possibly with string constants). Then,
    the string function associated with $t$
    is backwardable.
\end{lemma}
\label{sec:replaceable}
Proof of this lemma is folklore (and already exploited, e.g., in
\cite{popl19,z3noodler}). Other functions are also known to be forwardable
including the string-reverse function, \texttt{replaceAll(x,v,w)} --- which
replaces every occurrence of the
constant string $v \in \ialphabet^*$ in $x$ by the constant string $w \in
\ialphabet^*$ --- and more generally \emph{rational transducers} (e.g.\ see
\cite{berstel}), which are standard finite automata with input and output
(i.e.\ each transition has a label $(v,w) \in \ialphabet^* \times \ialphabet^*$,
and has the semantics of consuming $v$ in the input and then outputting $w$).
The same is no longer true for \texttt{replaceAll(x,v,y)} (i.e.\ where the
replacement string is a variable) \cite{CCHLW18}, and
two-way transducers (i.e.\ head of the
input/output tapes can move left/right), or equivalently, streaming
transducers \cite{AC10,AD11}. Interestingly, all of these aforementioned
functions are backwardable \cite{popl19}.

\OMIT{
\paragraph{Example of Non-Backwardable Function:} The function $f(x,y,z) = $ \texttt{replaceAll(x,y,z)} which takes strings $x,y,z$ and replaces each occurrence of $y$ in $x$ by $z$ is not backwardable.

In the following we will use that the following functions are forwardable and backwardable:

\begin{lemma}
	\label{lem:what_is_forwardable}
	If $f : \Sigma^* \to \Sigma^*$ is a rational function (so given by a rational transducer) then it is forwardable and backwardable.

	If $f : (\Sigma^*)^k \to \Sigma^*$ is defined as an (arbitrary) concatenation of its arguments (perhaps many times) then $f$ is backwardable.
	If additionally $f$ uses each of its arguments at most (in the concatenation) then it is forwardable.
\end{lemma}
All of the claims are folklore.
Note that the assumption that $f$ concatenates each of its arguments twice is essential:
$f(x) = xx$ applied on $\Sigma^*$ gives the language of all squares, which is well-known to be not regular (in fact: not context-free).
}

\begin{example} \label{ex:proof-tree}
	Given the string constraint:
	\[
	y = zu \land y = xx
	\land z \in b \land u \in a \ .
	\]
	we apply the proof rule \texttt{Fwd-Prop} to propagate the regular constraints $z \in b$ and $u \in a$ through the equational constraint $y = zu$, allowing us to deduce that $y \in ab$.

	Next, applying \texttt{Bwd-Prop} on $y \in ab$, we propagate this information backwards through $y = xx$. This leads to a contradiction in all branches of the proof.
	The proof tree illustrating this contradiction is shown in Figure \ref{fig:proof-tree}.

\begin{figure}
		
	\overfullrule=0pt
	\centering
	\scalebox{.75}{
		\parbox{\columnwidth}{
			\begin{prooftree}
				\AxiomC{}
				\RightLabel{$\texttt{[Close]}$}
				\UnaryInfC{$\Gamma, y = xx, y \in ba,  \mathbf{x \in b}, \mathbf{x \in a}$}

				\AxiomC{}
				\RightLabel{$\texttt{[Close]}$}
				\UnaryInfC{$\Gamma, y = xx, y \in ba, \mathbf{x \in \epsilon}, \mathbf{x \in ba}$}

				\RightLabel{$\texttt{[Bwd-prop]}$ on $y = xx$}
				\BinaryInfC{$y = zu, y = xx, z\in b, u \in a, \mathbf{y \in ba}$}
				\RightLabel{$\texttt{[Fwd-prop]}$ on $y = zu$}
				\UnaryInfC{$y = zu, y = xx,  z\in b, u \in a$}
	\end{prooftree}}}
\caption{Proof tree for Example \ref{ex:proof-tree}}\label{fig:proof-tree}
\end{figure}
\end{example}

%




\section{Completeness and Decidability}

\label{sec:guarantee}
\begin{figure}[h]
	\centering
	\begin{tikzpicture}[
		every node/.style={font=\small},
		every edge/.style={draw, thick},
		every loop/.style={min distance=10mm, out=130, in=50, looseness=8}
		]

		\node[red] (x) at (4,0) {$x$};
		\node[red] (y) at (2,0) {$y$};
		\node[red] (z) at (2,2) {$z$};
		\node[red] (u) at (4,1.5) {$u$};

		\draw[->] (y) -- (x) node[midway, below] {\scriptsize $(y = xx, \rightarrow)$};
		\draw[->] (z) -- (y) node[midway, left] {\scriptsize $(y = zu, \leftarrow)$};
		\draw[->] (u) to[bend right] (y) node[xshift = 1.3cm, yshift = 0.5cm] {\scriptsize $(y = zu, \leftarrow)$};

	\end{tikzpicture}
	\caption{Flow graph of $Eqs = \{y=zu,y=xx\}$ depicting flow of constraints.
		\label{fig:flowgraph}}
\end{figure}
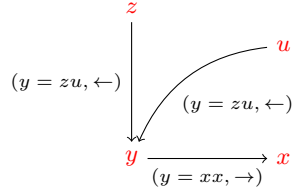
Here we provide a general completeness condition for RCP for
proving unsatisfiability.
By \emph{complete}, we mean that, for every unsatisfiable string constraint, there exists a closed RCP proof.
In general, this suffices to guarantee decidability
since if the constraint is satisfiable, then—by a standard argument in computability theory, based on the finiteness of the alphabet and the number of variables—one can enumerate all assignments of increasing length until a satisfying model is eventually found; otherwise, a closed RCP proof of unsatisfiability will be discovered.
%
Our completeness condition can
be understood in terms of finding a ``flow'' of regular constraints through
forwardable and backwardable functions. In a set
$Eqs =\{x_1 = f_1(y_1\ldots y_k),\ldots\}$ of
equational constraints, a variable $x$ that occurs \emph{only once} in $Eqs$
(say in equation $\phi \in Eqs$, on either left- or right-hand side)
is a potential ``source point'' of a flow. That is, the \emph{only} way $x$
can affect other variables is via the equation $\phi$. Therefore, by propagating
\emph{all} the regular constraints on $x$ through $\phi$ to other variables in
$\phi$, we may \emph{eliminate} $x$, as well as $\phi$, from $Eqs$, while preserving
satisfiability.
As an example, consider $Eqs = \{y = zu, y=xx\}$
along with regular constraints $z \in a^+ \wedge u \in b^+$. Since $z$ and $u$
occur only once and $f(z,u) := zu$ is forwardable, we could order the flow
of the regular constraints $z \in a^+ \wedge u \in b^+$ in the ``forward'' direction of
$y = f(z,u)$ resulting in the equisatisfiable constraint:
\[
	y = xx \wedge y \in a^+b^+.
\]
In turn, since $g(x) := xx$ is a backwardable function and $y$ occurs only
once, we could order the flow of the regular constraint $y \in a^+b^+$ in the
``backward'' direction resulting in the equisatisfiable constraint $\bot$.
That
is, we found the \emph{flow sequence} $(y=zu,\leftarrow),(y=xx,\rightarrow)$, where each pair depicts the equational constraint and the direction of propagation.
The propagation steps in a flow sequence can be visualized using a \emph{flow graph}. Given a flow sequence, we construct a directed graph where nodes represent variables. Each propagation rule $(\phi, d)$, with $\phi$ an equation like $y = f(x_1, \ldots, x_k)$, contributes edges from the variables on the source side of the propagation to those on the target side—i.e., from $x_1, \ldots, x_k$ to $y$ for a forward step, or from $y$ to $x_1, \ldots, x_k$ for a backward step. This flow graph helps visualize how information is propagated through equations. Notably, regular constraints do not influence the graph's structure, as it is determined solely by the equational relationships. An example of the flow graph for the constraint above is shown in Figure~\ref{fig:flowgraph}.
This intuition is formalized in the Marking Algorithm (Algorithm \ref{alg:marking}) and the
condition of \emph{orderable} set of equations, as given below.

\OMIT{
is the notion of \emph{effectively determined
variables}. To illustrate this, consider the

The intuition behind our
completeness condition is two-fold. First is the so-called \emph{deterministic
equation}. These are

an ability to ``order'' a
}

\SetKwComment{Comment}{/* }{ */}
\begin{algorithm}
\caption{Marking Algorithm\label{alg:marking}}
\KwData{A set of equational constraints $Eqs =\{x_1 = f_1(y_1\ldots y_k),\ldots\}$}
\KwResult{A flow sequence on the equational constraints, or $\bot$ if the constraints cannot be ordered}
$\mathrm{remainingEqs} \gets Eqs$\;
$\mathrm{markedVars} \gets \emptyset$\;
$\mathrm{flowSeq} \gets \epsilon$\;
\Repeat{$\mathrm{markedVars}$ is unchanged}{
    \ForEach{variable $x$ appearing exactly once in $\mathrm{remainingEqs}$}
    {
        $\mathrm{markedVars} \gets \mathrm{markedVars} \cup \{x\}$
    }
    \ForEach{$\phi = (x = f(y_1, \ldots, y_k)) \in \mathrm{remainingEqs}$ }
    {
        \If{$x \in \mathrm{markedVars}$ and $f$ is backwardable}
        {
            $\mathrm{flowSeq}
                 \gets \mathrm{flowSeq}, (\phi,\rightarrow)$\;
            $\mathrm{remainingEqs} \gets \mathrm{remainingEqs} \setminus \{\phi\}$\;
        }
        \ElseIf{$y_1, \ldots, y_k \in \mathrm{markedVars}$ and $f$ is forwardable}
        {
            $\mathrm{flowSeq} \gets \mathrm{flowSeq}, (\phi,\leftarrow)$\;
            $\mathrm{remainingEqs} \gets \mathrm{remainingEqs} \setminus \{\phi\}$\;
        }
    }
}
return $\mathrm{flowSeq}$ if $\mathrm{remainingEqs} = \emptyset$ else $\bot$
\end{algorithm}

\subsection{Orderable Constraints}

Given a set of equational constraints $Eqs$, a \emph{flow sequence} is a sequence of pairs $(\phi, d)$ where $\phi \in Eqs$ and $d \in \{\leftarrow, \rightarrow\}$.
Each such pair is called a \emph{propagation rule}.
A pair $(\phi, \leftarrow)$ indicates that the rule \texttt{Fwd-Prop} should be applied against the equational constraint $\phi$.
Conversely, $(\phi, \rightarrow)$ indicates a use of the \texttt{Bwd-Prop} rule.

\begin{definition}
    A set $Eqs$ of equations is said to be \emph{orderable} if, on input $Eqs$,
Algorithm~\ref{alg:marking} outputs a flow sequence. A string constraint $\varphi$
is orderable if the equational part is
    orderable.
\end{definition}
It can be easily checked that $\{y=zu,y=xx\}$ is orderable. An example of
a set of equations that is not orderable is $\{y=zu,y=xx,y=uv\}$; the
reason being that $u$ occurs on the right hand side of two equations,

The main result of this section is that the flow sequence returned by the Marking Algorithm gives a strategy to apply the forward and backward propagation
in RCP that is sound and complete.

\begin{theorem}
	\label{thm:marking_sound_an_complete_proof}
	Given a set of orderable constraints
	applying the forward propagation and backward propagation by RCP according
	to the flow sequence returned by Algorithm~\ref{alg:marking}
	is sound and complete.
	In particular, if such a string constraint is not satisfiable,
	the proof system will produce a proof with all closed branches.
\end{theorem}

In particular, if the instance of orderable constraints is not satisfiable,
RCP can produce a proof of that.

\begin{corollary}
	\label{cor:sound_complete_orderable}
	Every unsatisfiable orderable constraint admits a proof in RCP
	system with the following rules: \textnormal{\texttt{Close}}, \textnormal{\texttt{Intersect}}, \textnormal{\texttt{Fwd-Prop}}, and \textnormal{\texttt{Bwd-Prop}}.
\end{corollary}

\maineasychair{

To prove Theorem~\ref{thm:marking_sound_an_complete_proof}
we use the following Lemmas~\ref{lemma:fwd} and~\ref{lemma:bwd},
which show that the removal of the equational constraint by Algorithm~\ref{alg:marking} yields an equisatisfiable string constraint, in cases of forward and backwards propagation, respectively.
Note that the proof system does not actually remove the equations.

\begin{lemma}\label{lemma:fwd}
	If $z = f(y_1,\ldots, y_n)$ 
	is an equational constraint in a string constraint and $y_1,\ldots, y_n$ have exactly one occurrence in the equational part of string constraint
	and $f$ is forwardable then the string constraint obtained by propagating the constraints on $y_1,\ldots, y_n$ forward and removal of the equation
	are equisatisfiable.
	Formally
	\begin{align*}
		&\exists z \exists y_1,\ldots, y_n \;\phi(\bar{x}, z) \land z \in L_z \land z = f(y_1,\ldots ,y_n) \land \bigwedge_{i=1}^n y_i \in L_i\\
		\equiv\; & \exists z \exists y_1,\ldots, y_n  \; \phi(\bar{x}, z) \land z \in L_z \land z \in f(L_1,\ldots L_n) \land \bigwedge_{i=1}^n y_i \in L_i
	\end{align*}
\end{lemma}

\begin{lemma}\label{lemma:bwd}
	If $z = f(y_1,\ldots, y_n)$ is an equational constraint in a string constraint and $z$ has no other occurrence in equational part of the string constraint,
	then the string constraint obtained by propagating the regular constraints on $z$ backwards and removal of the equation
	are equisatisfiable.
	Formally:
	\begin{align*}
		&\exists z \exists y_1,\ldots, y_n \;\phi(\bar{x}, y_1,\ldots, y_n) \land z \in L_z \land z = f(y_1,\ldots ,y_n)
		\\
		\equiv  \bigvee_{j=1}^k &\exists z \exists y_1,\ldots, y_n \; \phi(\bar{x},  y_1,\ldots, y_n)
		\land z \in L_z \land \bigwedge_{i=1}^n y_i \in L_{i,j}
	\end{align*}
	where $f^{-1}(L_z) = \bigcup^k_{j=1} L_{1,j} \times \dots \times L_{n,j}$.
\end{lemma}

\begin{proof}[Proof of Theorem~\ref{thm:marking_sound_an_complete_proof}]
By Lemma~\ref{lem:RCP_is_sound} RCP is sound in general,
so in particular it is sound for some specific order of applying the proof rules.	

	By assumption that the constraints are orderable we obtain that Algorithm~\ref{alg:marking}
	terminates with a~flow sequence that gives an order on equational constraints
	(and information, whether we should propagate forwards or backwards).
	We will propagate them in this order.

	Note, that applying the backwards propagation creates many branches in the proof tree
	and when constructing the proof tree we need to take care of each open branch.
	However, we treat all open branches in the same way:
	we will apply the same rule to each open branch (and implicitly close branches on which contradiction was obtained).
	We apply the intersection rule for each variable, before each of the propagations (on each open branch), without mentioning it explicitly.	
	
	Observe that the proof rules do not modify nor remove the equational constraints.
	Therefore, the set of equational constraints is the same in each sequent, except for the axioms.
	However, as a tool of the construction, we mark some equational constraints in the open branches,
	and as an invariant:
	\begin{itemize}
		\item in each of the open branches the same set of equations is marked and it is the same
		as the set of equations $\mathrm{remainingEqs}$ from Algorithm~\ref{alg:marking} at the corresponding step
		\item after application of a rule, which changes the sequent $\Gamma$ to $\{\Gamma_i\}_{i=1}^k$,
		the constraint obtained by removing the unmarked constraints from $\Gamma$ is equivalent to the (disjunction of) constraints obtained by removing
		the unmarked constraints from each $\Gamma_i$.
	\end{itemize} 
	In particular, a variable occurs once in a marked equational constraints in some open sequent
	if and only if it occurs once in a marked equational constraints in each of the open sequents.

	The proof follows by applying the propagation according to the order returned by Algorithm~\ref{alg:marking}.
	Note that the above invariants are clearly satisfied at the beginning,
	when the root is the open sequent and we set all equations as marked
	and the Algorithm~\ref{alg:marking} has not removed any of the equational constraints.

	First, consider \texttt{Intersect} rule, which rewrites sequent $\Gamma$ as $\Gamma'$.
	We mark exactly the same set of equational constraints in $\Gamma'$ as they are marked in $\Gamma$.
	Also none equational constraints were changed. Hence the first invariant clearly holds.
	For the second invariant note that we replaced some regular constraints $x \in e_1, \ldots , x \in e_k$
	with $x \in e$, where $L(e) = \bigcap_{i=1}^k L(e_i)$,
	so clearly those are equivalent and so the second invariant holds as well.

	Consider the next equational constraint for the propagation (in the sequence returned by Algorithm~\ref{alg:marking}).
	If we are to propagate forwards for $z = f(y_1, \ldots, y_n)$,
	then variables $y_1, \ldots, y_n$ appear only once in the set of equational constraints in $\mathrm{remainingEqs}$ 
	and so in each of the open sequents they appear exactly once in the marked equations
	(by the invariant).
	Hence the string constraint (in an open sequent and restricted to marked equational constraints) is of the form as in Lemma~\ref{lemma:fwd},
	so of a form
	\[\Gamma, \Gamma', y_1 \in e_1, \dots, y_n \in e_n, z = f(y_1,\dots, y_n) \enspace ,\]
	where $\Gamma'$ are the unmarked equational constraints.
	Let the corresponding formula be 
	\begin{equation*}
	\exists z \exists y_1,\ldots, y_n \;\phi(\bar{x}, z) \land z \in L_z \land z = f(y_1,\ldots ,y_n) \land \bigwedge_{i=1}^n y_i \in L_i \enspace ,
	\end{equation*}
	where $\phi(\bar{x}, z)$ is the concatenation of all constraints in $\Gamma$.
	Using the Lemma~\ref{lemma:fwd} we can remove the equational constraint obtaining an equisatisfiable formula
	\begin{equation}
	\label{eq:fwd_prop_formula}
	\exists z \exists y_1,\ldots, y_n  \; \phi(\bar{x}, z) \land z \in L_z \land z \in f(L_1,\ldots L_n) \land \bigwedge_{i=1}^n y_i \in L_i \enspace .
	\end{equation}
	

	On the other hand, the proof system applies the \texttt{Fwd-Prop} rule, i.e.\
	\[
	\prftree[r]{
	}{
		\Gamma, \Gamma', z \in e, z = f(y_1,\dots, y_n), y_1 \in e_1,\dots, y_n \in e_n
	}{
		\Gamma, \Gamma', y_1 \in e_1, \dots, y_n \in e_n, z = f(y_1,\dots, y_n)
	} \enspace ,
	\]
	where $e$ denotes a regular expression such that $L(e) = f(L(e_1), \ldots, L(e_n))$.
	In the new sequent we additionally unmark $z = f(y_1 \ldots, y_n)$.
	Clearly, in all open branches the same sets of equational constraints
	are marked (as we unmark exactly the same equation $z = f(y_1 \ldots, y_n)$)
	and this the set that remain in Algorithm~\ref{alg:marking}. 
	
	It is left to observe that the sequent without the unmarked constraints $\Gamma, z \in e, y_1 \in e_1,\dots, y_n \in e_n$
	corresponds exactly to formula from~\eqref{eq:fwd_prop_formula},
	which we already know is equivalent to formula corresponding to old sequent without marked constraints.
	Hence, the second invariant was shown, which ends the proof for \texttt{Fwd-Prop}.
	
%
%
%
%
%
%
%

	So consider the case when the equational constraint $z = f(y_1, \ldots, y_n)$ is scheduled for backwards propagation.
	Then $z$ appears only once in the set of equational constraints in $\mathrm{remainingEqs}$ 
	and so in each of the open sequents it appears exactly once in the marked equations
	(by the invariant).
	Fix an open sequent $	\Gamma, \Gamma', z \in e, z = f(y_1,\dots, y_n)$, where $\Gamma'$ contains all unmarked equation constraints.
	Let the formula corresponding to $\Gamma, z \in e, z = f(y_1,\dots, y_n)$ be
	\[
	\exists z \exists y_1,\ldots, y_n \;\phi(\bar{x}, y_1,\ldots, y_n) \land z \in L_z \land  z = f(y_1,\ldots ,y_n) \enspace 
	\]
	where $L_z = L(e)$.
	Then it is of the form as in Lemma \ref{lemma:bwd}
	and we can remove the equational constraint obtaining an equisatisfiable formula:
	\begin{equation}
	\label{eq:prop_bwd_formula}
	\bigvee_{j=1}^k \exists z \exists y_1,\ldots, y_n \; \phi(\bar{x},  y_1,\ldots, y_n)
	\land z \in L_z \land \bigwedge_{i=1}^n y_i \in L_{i,j}
	\end{equation}
	

	On the other hand, the proof system applies the \texttt{Bwd-Prop} rule, i.e.\
	\[
   \prftree[r]{
}{
	\left\{
	\begin{array}{c}
		\Gamma, \Gamma',
		z \in e,
		z = f(y_1,\dots, y_n), \\
		y_1 \in e^i_1,\dots, y_n \in e^i_n
	\end{array}
	\right\}^k_{i=1}
}{
	\Gamma, \Gamma', z \in e, z = f(y_1,\dots, y_n)
}	\]
	where each $e_{i,j}$ is a regular expression such that $L_{i,j} = L(e_{i,j})$.
	In the new sequents we additionally unmark $z = f(y_1 \ldots, y_n)$.
	Clearly, in all open branches the same sets of equational constraints
	are marked (as we unmarked the same equation $z = f(y_1 \ldots, y_n)$)
	and this the set  $\mathrm{remainingEqs}$ in Algorithm~\ref{alg:marking}. 
	
	It is left to observe that the formula from~\eqref{eq:prop_bwd_formula}
	corresponds to the (disjunction of) new sequents without unmarked equational constraints
	$\left \{
		\Gamma, z \in e, y_1 \in e^i_1,\dots, y_n \in e^i_n
	\right\}^k_{i=1}$.
	Which ends the proof for all rules.
%
%

	As Algorithm~\ref{alg:marking} eventually removes all equations,
	at this step all open branches in the proof have all equational constraints unmarked,
	by the invariant.
	At each open branch we apply the \texttt{Intersect} rule for each variable.
	If the intersection is empty for some variable then we close the branch.
	If at some branch all intersections are non-empty,
	then we claim that initial formula is satisfiable:
	at this sequent we can find a substitution satisfying the regular constraints
	and there are no marked equational constraints in this sequent.
	Using the second invariant we conclude that on each sequent on the path to the root the constraint obtained by removing the unmarked constraints
	are satisfiable:
	the only needed observation is that if some sequent $\Gamma_j$ is satisfiable,
	then also $\left \{  \Gamma_i \right \}_{i=1}^k$ is satisfiable.
	In particular, in the root sequent there are only marked equations,
	so the whole input constraint is satisfiable,
	contradiction with the assumption that it is unsatisfiable.
	Thus, if the initial sequent is not satisfiable then we obtain that all branches are eventually closed when we follow an order of propagations returned by Algorithm~\ref{alg:marking},
	which shows that the proof rules are complete for unsatisfiable formulas.
	Note that this does not immediately yield a solution of the original string constraint
	for satisfiable constraints,
	as only equisatisfiability is guaranteed.
\end{proof}

}{}

\subsection{Straight-Line Constraints are Orderable}

The straight-line fragment of string constraints in~\cite{popl19} corresponds to constraints in which the equational part can be ordered with a flow sequence of the form $(\phi_1, \rightarrow), \ldots, (\phi_n, \rightarrow)$.
The fragment permits all string functions that are backwardable, which is called in~\cite{popl19} as a RegInvRel condition.
This includes a rich class of functions such as the transductions and \texttt{replaceAll} functions described in Section~\ref{sec:replaceable}.\footnote{The straight-line fragment also permits string relations that can be effectively expressed as recognizable relations, which can be reduced to regular constraints.}

\begin{corollary}
    RCP is sound and complete on straight-line constraints.
    In particular, every unsatisfiable straight-line constraint admits a proof in the RCP system with the following rules: \textnormal{\texttt{Close}}, \textnormal{\texttt{Intersect}}, and \textnormal{\texttt{Bwd-Prop}}.
\end{corollary}

\subsection{Chain-free Constraints are Orderable}
We show that orderable constraints subsume the class of chain-free constraints~\cite{chain19}.
In fact, when we restrict to functions allowed in the chain-free fragment then those two classes coincide;
in particular: unsatisfiable chain-free constraints admit a proof in RCP.

\begin{theorem}
		\label{thm:chain_free_is_orderable}

       A string constraint given by string terms and rational transducers,
       is chain-free if and only if it is orderable.
\end{theorem}


\begin{corollary}
	\label{cor:sound_complete_chain_free}
	RCP is sound and complete on chain-free constraints.
	In particular, every unsatisfiable chain-free constraint
	admits a proof in the RCP system with the following rules: \textnormal{\texttt{Close}}, \textnormal{\texttt{Intersect}}, \textnormal{\texttt{Fwd-Prop}}, and \textnormal{\texttt{Bwd-Prop}}.
\end{corollary}

Below we define chain-free constraints and show Lemmata needed for the proof of Theorem~\ref{thm:chain_free_is_orderable}.

\maineasychair{

\paragraph{Chain-free fragment}

The chain-free fragment of string constraints~\cite{chain19} is defined via a splitting graph:
Consider a string constraint $\psi = \bigwedge^n_{i=1} \phi_i$,
where each $\phi_j$ is of a form $t_{2j-1} = \mathcal{T}(t_{2j})$,
where $\mathcal T$ is a rational function (so given by a rational transducer) and $t_{2j-1},t_{2j}$ are string terms;
in particular this includes an equation $t_{2j-1} = t_{2j}$.
Each term $t_i$ is a concatenation of variables $x_{i,1}, \dots, x_{i,n_j}$ and constants,
one variable can occur many times in the concatenation.
We construct the splitting graph as follows:
\begin{description}
	\item[Nodes] The graph contains nodes $\{(j,i) | 1 \leq j \leq 2n, 1 \leq i \leq n_j\}$.
	The node $(2j-1, i)$ represents the $i$-th term on the left-hand side, while node $(2j, i)$ represent the $i$-th term on the right-hand side and they are labelled with the corresponding variables.
	\item[Edges] There is an edge from a node $p$ to node $q$ if there exists a node $p'$ (different from $q$) such that
	\begin{itemize}
		\item $p$ and $p'$ represents the nodes on opposite sides of the same constraint (say $\phi_i$), and
		\item $p'$ and $q$ have the same label.
	\end{itemize}
\end{description}
A chain in the graph is a sequence of edges of the form $(p_0,p_1),(p_1,p_2),\ldots, (p_n,p_0)$.
A splitting graph is chain-free if it has no chains;
a set of equational string constraints is chain-free if its splitting graph is chain-free;
a string constraint is chain-free if its equational part is chain-free.

Note that the original definition of chain-free fragment allowed rational relations,
so $t_{2j-1} \in \mathcal{T}(t_{2j})$ (or $\mathcal{T} (t_{2j-1}, t_{2j})$, depending on the preferred syntax);
our approach and methods generalize to this case, yet all relevant applications
and examples in the SMT-LIB benchmarks use rational functions only,
so for simplicity of presentation we use only rational functions.
On the other hand, the original algorithm showing the decidability of chain-free constraints~\cite{chain19}
assumed that the rational transducer is length-preserving,
which is not needed in the case of orderable constraints (we omit this assumption for ease of presentation).

}{}

\begin{lemma}
	\label{lem:marking_chain_free}
	A chain-free constraint $\psi$ is orderable.
\end{lemma}

\begin{proof}
	As regular constraints do not affect whether the set of constraints is chain free, nor whether it is orderable,
	we ignore them for the purpose of the proof.
	
	By Lemma~\ref{lem:normal_form_chain_free} we can assume that the string constraint is in normal form,
	let it be $\psi := \bigwedge^m_{j=1} \phi_j$ and $\phi_j := x_j = f(y_1, \ldots, y_{k_j})$ or $x_j = \mathcal{T}(y)$,
	where $f$ is an (arbitrary) concatenation of its variables and $\mathcal{T}(y)$ is a rational transducer.
	Note that in general $f$ does not need to be forwardable.
	Assume that the string constraint $\psi$ is chain-free, thus the splitting graph defined by $\psi$ has no cycles.
	For the sake of contradiction, assume that Algorithm \ref{alg:marking} returns $\bot$.
	That is, all equational constraints in $Eqs$ are eliminated.
	Observe that merging all nodes $(2j,i)$ over various $i$, so all nodes corresponding to a single r.h.s., does not create a cycle:
	all nodes $(2j,i)$ over various $i$ have the outgoing edges to exactly the same nodes, so if there is an incoming edge to $(2j,i)$ and outgoing from $(2j,i')$ to $p$
	then there is also outgoing from $(2j,i)$ to $p$.

	So let us merge all nodes corresponding to one r.h.s.\ to one node.
	Consider an equational constraint, say $\phi_j$.
	Then the variable on its left-hand side, say $x$, is not marked, as otherwise $f$ (or $\mathcal{T}$) is backwardable and so the equational constraint would have been removed.
	Hence this variable appears in some other equational constraint, say $\phi_{j'}$.
	Then the node corresponding to the right-hand side of $\phi_j$ has an outgoing edge (to the node representing $x$ in $\phi_{j'}$).
	Moreover, the node representing the other side of $\phi_{j'}$ has an edge to the l.h.s.\ of $\phi_j$.
	Observe that at least one variable of the r.h.s.\ of $\phi_j$ is not marked: if all of them were marked and the equational constraint
	is $x = \mathcal{T}(y)$ then $\mathcal{T}$ is forwardable and so this equational constraints should be removed,
	if the equational constraint is $x = f(y_1, \ldots, y_k)$ then as all $y_i$ are marked, each occurs once in the concatenation defined by $f$, and so $f$ is forwardable,
	so it also should be removed.
	Thus, one of the variables on the r.h.s.\ of $\phi_j$ appears somewhere else;
	using the same argument we show that the l.h.s.\ of $\phi_j$ also has at least one outgoing edge and the r.h.s.\ has at least one incoming edge.
	Thus, each node has an outgoing edge and an incoming edge,
	therefore the simplified splitting graph has a cycle, and so also the splitting graph has a cycle.

	Note that the observation that if all variables of $f$ are marked then $f$ is forwardable is crucial for the argument.
\end{proof}

\begin{lemma}
	\label{lem:chain_free_marking}
	If a string constraint given by string terms and rational transducers is orderable then it is chain-free.
\end{lemma}

\begin{proof}
	As in Lemma~\ref{lem:marking_chain_free} we can ignore the regular constraints.
	
	Note that if $\mathrm{remainingEqs}$ is not chain-free then $\mathrm{remainingEqs} \cup \{\phi\}$ is not chain-free:
	each edge in splitting graph for $\mathrm{remainingEqs}$ exists in the splitting graph for $\mathrm{remainingEqs} \cup \{\phi\}$.
	
	We show that whenever the Marking Algorithm removes an equational constraint $\phi$ from $\mathrm{remainingEqs} \cup \{\phi\}$,
	then in the splitting graph this corresponds to a removal of some nodes without incoming edges and some nodes without outgoing edges together with the edges leading from and to them, respectively, which clearly cannot remove nor create a cycle.
	
	So suppose that $\phi$ is an equational constraint $x = f(y_1, \ldots, y_k)$ or $x=\mathcal{T}(y_1)$ is removed because $x$ is marked.
	Then the node corresponding to $x$ has no incoming edges (as there is no other copy of $x$) and each $y_i$ has no outgoing edges (as there is no other occurrence of $x$).
	
	Finally, observe that the edges that are missing in the splitting graph for $\mathrm{remainingEqs}$ with respect to splitting graph $\mathrm{remainingEqs} \cup \{\phi\}$ are the edges to nodes corresponding to $y_1, \ldots, y_k$ and edges from $x$.
	
	The analysis in the case when $y_1, \ldots, y_k$ are marked are done in a similar manner (this time there are no incoming edges to each of $y_1, \ldots, y_k$ and no outgoing edge from $x$).
\end{proof}

Lemmata~\ref{lem:marking_chain_free}, Lemma~\ref{lem:chain_free_marking} prove Theorem~\ref{thm:chain_free_is_orderable};
Corollary~\ref{cor:sound_complete_chain_free} follows from Corollary~\ref{cor:sound_complete_orderable}.


\section{Lower Bounds}
\label{sec:lower}

We now show that RCP paired with Nielsen's transform and the Cut rule 
(for regular constraints, described in detail below)
cannot solve some string constraints;
those string constraints of course are not chain-free.
The concrete example in question is
\begin{equation}
\label{eq:lower_bound_example}
xy xy = yxyx \; \land \; x \in a^+  \; \land \; y \in a^* b a^*
\end{equation}
It is easy to see that the constraints in~\eqref{eq:lower_bound_example} are not satisfiable, by looking at the first position of 
$b$ in any candidate solution.
Note that length-analysis can reduce this equation into two copies of
$xy = yx$, for which Nielsen's transform can detect unsatisfiability, even without the Cut rule or regular constraint propagation.

For the sake of completeness, we give proof system rules for the Cut rule and Nielsen transformation,
see Fig.~\ref{fig:proof-rules-cut-Nielsen}:
For the Cut rule we create two branches by adding an arbitrary regular constraint and its complement on a~variable.

\begin{figure}[tb]
\texttt{[Cut]}
\[
    \prftree[r]{
    	if $L(e')=  \Sigma^* \setminus L(e)$
    	}{
        \Gamma, x \in e \quad \Gamma, x \in e' 
    }{
        \Gamma
    }
\]
\texttt{[Nielsen]}
\[
\prftree[r]{}{
	\prfStackPremises
	{\Gamma[x/y], u[x/y] = v[x/y], \{y \in e\}_{e \in R_x \cup R_y}}
{\{	\Gamma[x/yx], x(u[x/yx]) = v[x/yx], x \in e_x^i, y \in e_y^i\}^k_{i=1}}
{\{	\Gamma[y/xy], u[y/xy] = y(v[y/xy]), x \in e_x^i, y \in e_y^i\}^\ell_{i=k+1}}
}
{	\Gamma, \{x \in e\}_{e \in R_x}, \{y \in e\}_{e \in R_y}, xu = yv  
}
\]
\caption{\label{fig:proof-rules-cut-Nielsen} Additional proof rules.}
\end{figure}

Nielsen's transform considers leading symbols of an equation and branches into three cases, depending on which symbol is a prefix of the other
(there is also a symmetric version for the suffix).
In Figure~\ref{fig:proof-rules-cut-Nielsen} the cases from top-to-bottom are: $x$ equals $y$, $y$ is a prefix of $x$, and finally $x$ is a prefix of $y$.
In particular, it makes substitutions such as $x \gets yx'$ (when $y$ is a prefix of $x$) and needs to compute the regular constraints
for $x'$.
After the substitution and computation of the constraints, $x'$ is renamed to $x$.
We use the notation $u[x/y]$ to denote that $y$ is substituted in place of $x$ in the term $u$.

The rule is somewhat similar to the rule for backward propagation,
in particular it has several branches.
This is not surprising, as for a fixed $y$ we can treat $yx'$ as function that prepends $y$ to its argument $x'$. We have a regular constraint ($e_x$) on its output and want to compute (propagate) the constraint
 $e_{x'}$ on its input.
In particular, we will replace the constraints $e_x, e_y$ with a family
$\{(e_y^i, e_x^i)\}$, one pair for each branch.
The actual formula for $\{(e_y^i, e_x^i)\}$ can be given.\footnote{
	We consider DFA $A$ recognizing $L(e_x)$ and consider the possible state $q$ reached by $A$ after reading $y$
	and require that $A$ starting in $q$ and reading $x'$ ends in a final state;
	alternatively, this can be handled using matrix transition functions or the semigroup-approach to regular languages.
}
However, for the purpose of the lower bound we do not need this exact formula
and we can in fact allow the proof system to overapproximate it --- 
it is enough to assume ``soundness of the substitution''.
Let $e_x, e_y$ be the intersection of all constraints on $x, y$,
i.e.\ $L(e_x) = \bigcap_{e \in R_x}  L(e), L(e_y) = \bigcap_{e \in R_y}  L(e)$, we assume
\begin{align}
	\label{eq:no_underapproximation}
	\exists x' \, x = yx' \land y \in L(e_y) \land x \in L(e_x) 
		&\implies
	\exists i \; y \in L(e^i_y) \land x' \in L(e^{i}_x) 
\intertext{i.e.\ we allow that the new constraints $e_y^i, e_x^i$ overapproximate the old one.
	However, we require a simple upper bound where $e_y^i$ only accepts values that also appear in $e_y$ and $e_x^i$ only accepts values that are correct for some valid value of $y$ (not necessarily accepted by $e_y^i$).
	I.e.\ we assume}
\label{eq:ey_over_approximate}
L(e_y^i) &\subseteq L(e_y)\\
\label{eq:ex_over_approximate}
L(e_x^i) &\subseteq \{v \: : \: \exists u \, uv \in L(e_x) \land u  \in L(e_y)\}
\end{align}

Note that in general Nielsen's rule (as given above) does not detect a solution in which initially some variable has an empty substitution.
However, such substitutions are forbidden by the regular constraints in~\eqref{eq:lower_bound_example}. 

We say that a set $A$ is co-finite in $B$, when $B \setminus A$ is finite.
Consider a condition:

\begin{enumerate}
\item if the regular constraint on $x$ defines a language $L_x$ and for $y$ a language $L_y$
then there is $m$ such that $L_x$ is co-finite in \akm{}
and there are languages  $L_{y,\ell}, L_{y,r}$ co-finite in \akm{}
satisfying $L_{y,\ell}bL_{y,r} \subseteq L_y$. \label{eq:condition}
\end{enumerate}

\begin{theorem}
	\label{thm:propagation_isnot_enough}
	Consider a proof system that uses: Nielsen's rule, regular constraint propagation and the Cut rule only and a proof for~\eqref{eq:lower_bound_example} using this system.
	If Condition~\ref{eq:condition} holds at some sequent, then after application of a proof rule it holds in at least one of the obtained sequents (perhaps for some other value of $m$).
\end{theorem}
%
%

In this way we obtain the main result of this section:
\begin{corollary}
	\label{cor:propagation_isnot_enough}
A proof system that uses: Nielsen's rule, regular constraint propagation and the Cut rule
only, cannot prove the unsatisfiability of~\eqref{eq:lower_bound_example}.
\end{corollary}
\begin{proof}
Condition \eqref{eq:condition} implies in particular that there are (infinitely many) words
satisfying the regular constraints, so in particular we cannot close a sequent satisfying the condition.

Clearly, condition~\eqref{eq:condition}	holds for the initial regular constraints
and so by Theorem~\ref{thm:propagation_isnot_enough}
there exists an infinite path of sequents, in particular, we cannot close the whole proof, which shows the claim.
\end{proof}

To prove Theorem~\ref{thm:propagation_isnot_enough} we use two auxiliary Lemmata:

\begin{lemma}
	\label{lem:regular_representation}
	Given regular sets $L_x \subseteq a^*, L_y \subseteq a^* b a^*$
	there is a number $m$ such that
	\begin{align*}
		L_x
		&=
		\bigcup_i  L_i & \text {where } L_i \in \left\{a^{k_i}, a^{k_i}\akm\right\} \text{ for each }i\\
		L_y
		&=
		\bigcup_i L_{i,\ell} b L_{i,r} & \text {where } L_{i,\ell} \in \left\{a^{k_{i,\ell}}, a^{k_{i,\ell}}\akm\right\}, L_{i,r} \in \left\{a^{k_{i,r}}, a^{k_{i,r}}\akm\right\} \text{ for each }i
	\end{align*}
	for appropriate values $k_i, k_{i,\ell}, k_{i,r}$ (over all $i$).
	In particular, each $ L_i,  L_{i,\ell}, L_{i,r}$
	is  either finite or co-finite in \akm.
	Furthermore, if such a representation exists for $m$, then it exists for each
	multiple of $m$.
	In particular, if we are given several such sets, then there is a representation of
	all of them for a common $m$.
\end{lemma}
In case of subsets of $a^*$ the proof follows from application of well-known characterizations of unary regular languages, say Chrobak normal form.
The proof for subsets of $a^*ba^*$ requires some additional simple arguments.

\begin{lemma}
	\label{lem:infinite_to_cofinite_lemma}
	If regular sets $L_x \subseteq a^*, L_y \subseteq a^* b a^*$
	are represented as in Lemma~\ref{lem:regular_representation}
	for a common $m$ and there are infinite sequences
	of words $x_i = a^{k_im} \in L_x$,
	$y_i = a^{k_{i,\ell}m}ba^{k_{i,r}m} \in L_y$
	where all $k_i$ are pairwise different,
	all $k_{i,\ell}$ are pairwise different
	and all $k_{i,r}$ are pairwise different,
	then $L_x$ is co-finite in \akm and there are $L_{y,\ell}, L_{y,r}$ co-finite in \akm such that $L_{y,\ell}bL_{y,r} \subseteq L_y$.
\end{lemma}
The proof follows from simple observations applied to Lemma~\ref{lem:regular_representation}:
for appropriate $m$, if a regular language has an infinite intersection with $\akm$
($\akm b \akm$) then it is in fact co-finite in it (for subsets of $a^* b a^*$ please consult the statement).

To prove Theorem~\ref{thm:propagation_isnot_enough} we apply
Lemma~\ref{lem:infinite_to_cofinite_lemma}: given a sequent
with regular constraints $e_x, e_y$ on $x, y$
we choose appropriate infinite sequences in the regular languages $L(e_x) \cap \akm,
L(e_y) \cap \akm b \akm $, which is possible by the assumption that they satisfy condition~\eqref{eq:condition},
and show that after a rule application at least one of the sequents has regular constraints
$e_x', e_y'$, which have infinitely many elements from the sequences,
and so by Lemma~\ref{lem:infinite_to_cofinite_lemma} they in fact satisfy condition~\eqref{eq:condition}.
The choice of sequences depends on the applied rule; say, for a Cut rule,
in which we replace the regular constraint $e_y$ with $e_y, e$ and $e_y, \overline e$
observe that infinitely many elements of the sequence in $L(e_y)$ are in at least one of
$L(e_y) \cap L(e)$ or $L(e_y) \cap \overline{L(e)}$.

\section{Experiments}\label{sec:expr}

For evaluation, we have implemented a new string solver, \solver{}, which is built on top of OSTRICH. OSTRICH relies on a robust foundation---including the BRICS automata framework~\cite{brics}, transducer support, and pre-implemented backward propagation---that enabled us to straightforwardly extend the solver with RCP.

To isolate the impact of RCP, we implemented it on top of the \emph{basic OSTRICH} configuration, which combines backward propagation with equation splitting. In doing so, we disabled other advanced techniques used in OSTRICH’s SMT-COMP portfolio, such as cost-enriched automata (CEA)~\cite{cea} and the algebraic data type (ADT)-based solver~\cite{adt}. As a result, our evaluation considers two configurations of OSTRICH:  
(1) the \emph{SMT-COMP portfolio version}, which integrates basic OSTRICH, CEA-OSTRICH, and the ADT-based solver; and  
(2) the \emph{basic OSTRICH}, on which we directly implement and evaluate RCP.

This setup allows us to measure the isolated benefits of regular constraint propagation and compare it both to the unmodified basic configuration and the full SMT-COMP portfolio variant.

We evaluate \solver{}, as well as cvc5 \cite{cvc5}, OSTRICH \cite{ostrichsmtcomp}, Z3 \cite{z3}, Z3-alpha \cite{z3alpha1}, and Z3-Noodler \cite{z3noodler}, using their latest versions configured as in SMT-COMP 2024 \cite{SMT-COMP2024}.
%
%
The experiments were conducted on a MacBook Pro with 16 GB of RAM, running macOS Sonoma 14.5. The system was powered by an Apple M3 chip. The timeout for each experiment was set to 60 seconds.

\subsection{Proof Search}
\label{sec:proof-search}

To transform the proof system (Figure~\ref{fig:proof-rules}) into a practical algorithm, one must define a strategy for applying the various proof rules. 
In principle, this can be done using Algorithm~\ref{alg:marking} for the case of orderable string constraints, as the algorithm outputs a sequence of propagation rules that can be executed to prove (un)satisfiability. 
While this provides a complete propagation order for formulas within the orderable fragment, it does not generalize to formulas outside this fragment. Moreover, our goal is to define a uniform procedure that applies broadly, even when completeness cannot be guaranteed.

As a baseline, we describe a simple ``fair'' scheduling strategy (see Appendix~\ref{sec:proof-strategies}) that is easy to implement and ensures that every applicable propagation rule is eventually considered. 
To improve on this baseline, we employ a priority-based refinement of fair scheduling that guides propagation using a cost heuristic.
Our solver introduces a bias toward inexpensive and potentially informative propagation steps, e.g., forward propagation from small automata or from functions with constant arguments while at the same time delaying applications that involve large automata or only yield coarse approximations. 
This heuristic assigns priorities to applications of forward/backward propagation using a weighted sum of factors, including:
\begin{itemize}
	\item \textbf{Concrete arguments:} high priority is given to rules whose input or output language is a ground string;
	\item \textbf{Information gain:} backward rules propagating from universal automata (e.g., $x \in \Sigma^*$) receive lower priority;
	\item \textbf{Exactness:} forward rules involving functions without exact post-images (e.g., \texttt{replaceAll} with symbolic patterns) are penalized;
	\item \textbf{Cost:} rules are penalized in proportion to the combined size of input and result automata;
	\item \textbf{Fairness:} newer rule applications are penalized, preferring those that have waited longer in the queue.
\end{itemize}

\subsection{Benchmarks} \label{sec:benchmarks}
We evaluate our solver on three benchmark sets. The first benchmark set consists of 2,000 formulas randomly selected from the QF\_S and QF\_SLIA divisions of SMT-LIB 2024 \cite{smtlib2024}. To ensure representativeness, we sampled proportionally to the distribution of benchmarks in the full SMT-LIB 2024 dataset, which contains approximately 100,000 benchmarks. This resulted in a subset of roughly 400 benchmarks from QF\_S and 1,600 from QF\_SLIA, preserving the overall ratio between the divisions.

Below, we provide descriptions for some of the different benchmark categories within the SMT-LIB:
\begin{itemize}
	\item \textbf{Regular Constraints + Word Equations}: Benchmarks from \textit{SyGuS}, \textit{Automatark}~\cite{automatark}, \textit{Woorpje-lu}~\cite{woorpje}, \textit{String Fuzz}~\cite{stringFuzz}, \textit{ReDoS Attack}~\cite{reDos} and \textit{Kepler}~\cite{Kepler}, which combine word equations with regular constraints and length.
	\item \textbf{String Constraints from Programs}: Extracted from concolic execution engines \textit{Conpy} \cite{conpy}, \textit{PyExZ3}~\cite{pyexz3} symbolic execution of Python programs~\cite{PyEx}, covering operations such as \texttt{replace}, \texttt{indexof}, \texttt{substr}.
	\item \textbf{String Rewrite}: Consists of benchmarks derived from axioms over string predicates and functions~\cite{strRewrite}.
	\item \textbf{WebApp}: A variation of \textit{SLOG} \cite{slog} using \texttt{replaceAll} instead of simple \texttt{replace}, derived from real web applications.
\end{itemize}

A key limitation of the SMT-LIB benchmark set is the relative scarcity of benchmarks containing \emph{complex string functions}. While modern programming languages such as \emph{JavaScript} heavily rely on operations like \texttt{replaceAll}, such functions are significantly underrepresented in the SMT-LIB benchmark releases.
For example, only \emph{0.4\%} of the benchmarks in the SMT-LIB 2024 release include \texttt{replaceAll}.

The second benchmark set consists of instances of \textbf{PCP}, specifically in the form PCP[3,3], meaning each instance contains three dominos, and each word on a domino has length three.
We generate 1,000 random instances over a binary alphabet, following the structure outlined in Example~\ref{ex:pcp}.

The third benchmark set is inspired by \textbf{bioinformatics}~\cite{mount2004bioinformatics,compeau2015bioinformatics} and models a reverse transcription process. In these benchmarks, an unknown RNA string $y$ is transformed into a DNA string by applying a sequence of \texttt{replaceAll} operations that simulate nucleotide base pairing. The solver must determine an RNA string such that, after these replacement operations, a given concrete DNA string is produced, while ensuring that the RNA string contains a specific pattern. We generate 1000 instances, with 500 instances being satisfiable (where a suitable $y$ exists) and 500 unsatisfiable. Each formula contains 4 \texttt{replaceAll} constraints, a constant DNA string of length 200, and a pattern string of length 15.

\newcommand{\cmark}{\ding{51}}  
\newcommand{\xmark}{\ding{55}}  
\begin{table}[t]
	\centering
	\caption{Solver results on the motivating examples in Section~\ref{sec:motivating-example}. \cmark = solved, \xmark = timeout or unknown.}
	\label{tab:motivating-examples}
	\begin{tabular}{lccc}
		\toprule
		Solver & String Sanitization & PCP & Bioinformatics \\
		\midrule
		\solver{}      & \cmark & \cmark & \cmark \\
		cvc5           & \xmark & \xmark & \xmark \\
		OSTRICH   & \xmark & \xmark & \cmark \\
		Z3             & \xmark & \xmark & \xmark \\
		Z3-alpha            & \xmark & \xmark & \xmark \\
		Z3-Noodler     & \xmark & \xmark & \xmark \\
		\bottomrule
	\end{tabular}
\end{table}

We also evaluated solver performance on the three illustrative formulas from Section~\ref{sec:motivating-example}: the string sanitization example, a simplified PCP tile instance, and a representative example from the bioinformatics benchmark set. Table~\ref{tab:motivating-examples} summarizes the results. \solver{} successfully solved all three instances, while existing solvers either timed out or returned \texttt{unknown}. The PCP formula is unsatisfiable, and only \solver{} was able to prove this. The normalization task lies outside the straight-line fragment and involves multiple transductions on the same variable; while OSTRICH supports transducers, it fails to solve this instance. By contrast, the bioinformatics example \emph{is} within the straight-line fragment and can be solved by both \solver{} and OSTRICH, but not by any other solvers.

\begin{table*}[t]
	\centering
\caption{Comparison of various solvers' performance on 2000 SMT-LIB'24 benchmarks, 1000 PCP benchmarks, and 1000 Bioinformatics benchmarks. For SMT-LIB'24, we report the number of satisfiable (Sat), unsatisfiable (Unsat), and unknown (Unknown) instances, along with total solving time (Time(s)), including timeouts. For PCP and Bioinformatics, we report the number of solved instances and total solving time. For all benchmarks, we conservatively assign a time of 60 seconds to each unknown result, even if the solver terminated earlier.}

	\label{tab:solver_results}
	\begin{tabular}{lcccccccc}
		\toprule
		\rowcolor{lightgray} 
		& \multicolumn{4}{c}{\textbf{SMT-LIB'24}} 
		& \multicolumn{2}{c}{\textbf{PCP}} 
		& \multicolumn{2}{c}{\textbf{Bioinformatics}} \\
		\rowcolor{lightgray} 
		& Sat & Unsat & Unknown & Time(s) 
		& Solved & Time(s)
		& Solved & Time(s) \\
		\midrule
		\solver{}    & 1071 & 728 & 201  & 15416.8  & 901 & 7308.7 & 1000 & 5532.6\\
		cvc5         & 1162 & 716 & 122  & 7954.3   & 0 & 60000.0      & 0 & - \\
		OSTRICH-BASE      & 833 & 653 & 514  & 32298.7  & 0 & 60000.0      & 1000 & 2596.2     \\
		OSTRICH-COMP      & 1009 & 733 & 258  & 22840.7  & 0 & 60000.0      & 1000 & 3911.0     \\
		Z3           & 1156 & 730 & 114  & 8286.0  & 0 & 60000.0      & 0 & 60000.0 \\
		Z3-alpha     & 1127 & 724 & 149  & 10681.2  & 0 & 60000.0      & 0 & 60000.0 \\
		Z3-Noodler   & \textbf{1236} & \textbf{749} & \textbf{15} & 797.0 & 0 & 60000.0      & 0 & 60000.0 \\
		\bottomrule
	\end{tabular}
\end{table*}

\subsection{Performance of Individual Solvers}
We begin by analyzing Table~\ref{tab:solver_results}, which presents the performance of all solvers across both benchmark sets.
We report results using different metrics for the two benchmark sets.
We focus on \emph{unsolved benchmarks} to highlight solver weaknesses.

For the first benchmark set, Z3-Noodler emerges as the best-performing solver, solving nearly all instances and leaving only \emph{{\raise.17ex\hbox{$\scriptstyle\mathtt{\sim}$}}1\%} of the benchmarks unsolved. It is closely followed by {Z3} and {cvc5}, both of which leave \emph{6\%} of the benchmarks unsolved. Next, {Z3-alpha} performs slightly worse, failing to solve \emph{7.5\%} of the benchmarks. \solver{} leaves  \emph{10\%} of the benchmarks unsolved.
{OSTRICH-COMP} fails to solve \emph{13\%} of the benchmarks, while {OSTRICH} fails on \emph{25\%}.

These results demonstrate that while \solver{} does not achieve top-tier performance, it remains competitive using simple strategies, effectively solving a significant portion of the benchmark set. Notably, it reduces the number of unsolved benchmarks from \emph{25\%} (roughly 500) with \emph{OSTRICH} to \emph{10\%} (around 200), a \textbf{60\% reduction} in the number of failures.

A different trend emerges in the second and third benchmark sets, which consist of \emph{1,000 PCP benchmarks} and \emph{1,000 bioinformatics benchmarks}. In contrast to the SMT-LIB'24 benchmarks, most solvers fail to solve any instances from these categories. \solver{} is the only exception, solving nearly all PCP instances (with only around \emph{10\%} remaining unsolved). While OSTRICH is unable to solve the PCP benchmarks, it manages to handle the bioinformatics benchmarks.
This disparity for OSTRICH can be attributed to the underlying structure of the benchmarks: while the PCP benchmarks involve a more complex structure, the bioinformatics benchmarks fall into the straight-line fragment, for which OSTRICH is effective.
These results highlight that outside of the SMT-LIB benchmark sets, many challenging benchmarks remain unsolved by existing solvers. However, our propagation-based approach solves a significant portion of these previously unsolved benchmarks, demonstrating that constraint propagation can be crucial for tackling problems beyond the reach of current solvers.

\subsection{Impact of Enhancing Solver Performance with \solver{}}

\begin{table*}[t]
	\centering
	\caption{Benchmark Results for Different Solver Combinations. For each benchmark suite, we report the number of benchmarks left unknown by the portfolio (U), the additional benchmarks solved by \solver{} (C), and the resulting Improvement (I), defined as \(\displaystyle \frac{\text{C}}{\text{C}+\text{U}}\). The last four rows below show additional RCP contribution when added to solvers already combined with OSTRICH-BASE (OB).}
	\label{tab:solver_results-portfolio}
	\begin{tabular}{l|ccc|ccc|ccc}
		\toprule
		\rowcolor{lightgray} 
		& \multicolumn{3}{c}{\textbf{SMT-LIB'24 (2000)}} 
		& \multicolumn{3}{c}{\textbf{PCP (1000)}} 
		& \multicolumn{3}{c}{\textbf{Bioinformatics (1000)}} \\
		\rowcolor{lightgray} 
		Solver Combination 
		& U & C & I 
		& U & C & I  
		& U & C & I  \\
		\midrule
		\solver{} 
		& 201  & --   & -- 
		& 99   & --   & -- 
		& 0    & --   & -- \\
		\hline
		cvc5 + \solver{}      
		& 61   & 61   & 0.50 
		& 99   & 901  & 0.90
		& 0    & 1000 & 1.00 \\
		OSTRICH-COMP + \solver{}   
		& 68   & 197  & 0.74 
		& 99   & 901  & 0.90
		& 0    & 0    & 0.00 \\
		Z3 + \solver{}        
		& 85   & 29   & 0.25
		& 99   & 901  & 0.90 
		& 0    & 1000 & 1.00 \\
		Z3-alpha + \solver{}  
		& 87   & 65   & 0.43
		& 99   & 901  & 0.90 
		& 0    & 1000 & 1.00 \\
		Z3-Noodler + \solver{} 
		& \textbf{11}   & 4    & 0.27
		& 99   & 901  & 0.90 
		& 0    & 1000 & 1.00 \\
		\hline
		(cvc5 + OB) + \solver{}      
		& 61   & 2   & 0.03 
		& 99   & 901  & 0.90 
		& 0    & 0    & 0.00 \\
		(Z3 + OB) + \solver{}        
		& 85   & 5   & 0.06
		& 99   & 901  & 0.90 
		& 0    & 0    & 0.00 \\
		(Z3-alpha + OB) + \solver{}  
		& 87   & 8   & 0.08 
		& 99   & 901  & 0.90 
		& 0    & 0    & 0.00 \\
		(Z3-Noodler + OB) + \solver{} 
		& \textbf{4}   & 0  & 0.00
		& 99   & 901  & 0.90 
		& 0    & 0    & 0.00\\
		\bottomrule
	\end{tabular}
\end{table*}

Table~\ref{tab:solver_results-portfolio} presents the performance of \solver{} and its impact when combined with various solvers. Specifically, for each solver, we compute the \emph{virtual best portfolio} (VBP), which considers the union of benchmarks solved either by the solver alone or by \solver{}. This is a particularly useful measure of impact, as many practical approaches for solving string constraints rely on a portfolio of solvers running in parallel. Evaluating the VBP provides insight into whether including a given solver in such a portfolio is beneficial, thereby helping to assess the practical value of \solver{} when integrated into a multi-solver strategy.
We omit \emph{OSTRICH-BASE} as base solver from this evaluation, as its technique is very close to \solver{}, and the previous experiments have already demonstrated the superior performance of \solver{} over this baseline.

We assess the unique contribution of \solver{}, defined as the number of benchmarks solved by \solver{} that the individual solver could not solve on its own. For each benchmark suite, we report three key quantities: the number of benchmarks left unsolved by the portfolio solver, the contribution provided by \solver{}, and the resulting improvement. The improvement is computed as
$(\frac{\text{contribution}}{\text{unsolved} + \text{contribution}})$
which quantifies the relative gain achieved by including \solver{} in a portfolio. This metric effectively demonstrates the impact of \solver{} in complementing existing solvers.

First, it is immediately apparent that for the PCP benchmarks, no other solver is able to solve any instances, making \solver{} solely responsible for a contribution of 901 solved benchmarks in this set. As a result, all solver combinations with \solver{} show an improvement of approximately 90\% on the PCP benchmarks. In addition, for the bioinformatics benchmarks the improvement is 100\% for all solver combinations except for OSTRICH + \solver{}, as OSTRICH can already solve the entire benchmark suite (yielding an improvement of 0\%).

For the SMT-LIB benchmark set, the contribution of \solver{} varies noticeably across the different base solvers. For instance, when combined with cvc5, \solver{} contributes an additional 61 solved instances from 122 unknown benchmarks, resulting in an improvement of  50\%. 
Meanwhile, combinations with Z3 and Z3-alpha show intermediate improvements of roughly 25\% and 43\%, respectively.
These improvements follow the same pattern, primarily benefiting from enhanced handling of regular constraints, word equations, and length reasoning across various benchmark categories.
In contrast, the combination of OSTRICH with \solver{} yields a significantly higher contribution, solving an extra 197 out of 258 unknown cases—translating to an improvement of about 74\%.\footnote{This also provides a strong indication of the potential improvement in \emph{OSTRICH-COMP} if the \emph{OSTRICH-BASE} component were replaced by our RCP implementation—reducing the number of unsolved instances to just 3\%, and ranking just behind \emph{Z3-Noodler}.}

The improvements are most pronounced in the String Constraints from Programs and String Rewrite benchmarks.
Its propagation strategy proves less effective on benchmarks involving more complex string functions and predicates, leading to weaker performance in those cases.
Finally, Z3-Noodler sees a more modest contribution (4 additional instances from 11 unknown, or 27\% improvement). 
This outcome is unsurprising, as Z3-Noodler also employs RCP as part of its strategy and solves already most of the benchmarks in this suite. The remaining unsolved benchmarks fall into two categories: \emph{one} originates from String Constraints from Programs, while \emph{three} belong to the WebApp category, which relies on \texttt{replaceAll}. Since this operation is currently unsupported in Z3-Noodler, these results indicate that while its propagation techniques are effective, its coverage of string functions in the SMT-LIB 2.7 format remains incomplete.

Notably, while Z3-Noodler showed only a minor absolute improvement on the SMT-LIB benchmarks, this can largely be attributed to the \emph{underrepresentation of complex string functions} such as \texttt{replaceAll} in the SMT-LIB dataset. Since Z3-Noodler already excels at solving the predominant types of string constraints in SMT-LIB, its gains from \solver{} are limited. However, as observed in the WebApp benchmarks, where \texttt{replaceAll} is more frequently used, \solver{} provided improvements.

The last four rows of Table~\ref{tab:solver_results-portfolio} show the additional contribution of \solver{} when added to portfolios already containing OSTRICH-BASE. As expected, the improvements on SMT-LIB'24 benchmarks are marginal (ranging from 0\% to 8\%), due to significant overlap between the constraints solved by \solver{} and OSTRICH-BASE. This highlights a natural trend: as more solvers are combined in a portfolio, the marginal contribution of each individual technique diminishes. Nonetheless, the value of RCP becomes more apparent on harder or previously unaddressed benchmarks. In particular, while OSTRICH-BASE fails to solve any PCP instances, adding RCP recovers 901 solved cases across all configurations, demonstrating its unique strength in handling expressive string constraints beyond SMT-LIB.

%

\section{Related Work} \label{sec:related-work}
String solving has been extensively studied 
in the past decade or so, e.g., see the survey \cite{string-survey}. 
Our work contributes to both the theoretical and practical aspects of string solving. In this section, we organize the related literature into two complementary parts. First, we review the theoretical foundations that underpin string constraint solving—ranging from early results on decidability.
Second, we discuss practical string solvers, highlighting the evolution of implementations from early prototypes to state-of-the-art tools. 

\paragraph{Theoretical Foundations.}

Early work on the satisfiability of word equations established decidability with significant complexity results. Makanin~\cite{Makanin1977} and later Plandowski~\cite{plandowski1999} showed that the problem is decidable and lies in PSPACE, while subsequent work by Je\.{z}~\cite{jez2016} introduced the recompression technique—a simpler method that streamlined earlier approaches. In practice, Nielsen's transformation~\cite{Nielsen1917} has emerged as an effective technique for handling quadratic word equations, maintaining the same theoretical complexity bounds. Although the lower bound is known to be NP, establishing tight bounds remains challenging.

The decidability of word equations with linear length constraints such as $|x| = |y|$ remains an open problem.
However, for quadratic word equations, some positive results have been obtained by combining Nielsen's transformation with counters to track variable lengths~\cite{quad21}. 
Moreover, even simple extensions incorporating transducers render the problem undecidable, which motivated the definition of the straight-line fragment~\cite{popl16}. 
This fragment restricts constraints so that each variable is assigned at most once—an intuitively appealing condition that nonetheless captures a rich class of problems. While the computational complexity of the straight-line fragment is EXPSPACE in general, it drops to PSPACE for constraints of small dimensions (i.e., when the structure of the constraints is sufficiently limited) and remains decidable even when augmented with additional constraints such as length, letter counting, and disequality. 
Recent work has further refined these ideas: a comprehensive proof system for string constraints that captures the straight-line fragment was introduced in~\cite{popl22}, providing a formal framework that informs our approach, and the straight-line condition was extended to the theory of sequences in~\cite{cav23} by applying constraint propagation techniques on parametric languages rather than on regular languages (following a similar approach to that in~\cite{popl19}).

Building on the idea of limiting variable assignments to ensure decidability, the chain-free fragment~\cite{chain19} extends the decidable class of string constraints by allowing variables to be assigned more than once, provided that these assignments do not form a chain. 
Although this extension builds on the concepts underlying the straight-line fragment, the techniques employed are fundamentally different.
In the chain-free fragment, a splitting graph is introduced to systematically manage the interactions between variables—particularly when a variable appears multiple times on the same side of an equation—thus enabling a richer set of constraints while preserving decidability.

In this context, our work shows that the proof system from~\cite{popl22}, specifically the rules for regular constraint propagation, is sufficient to subsume the chain-free fragment, providing a novel algorithmic characterization of it.

\paragraph{String Solvers.}
Over the past decade, a wide range of string solvers have emerged, employing techniques such as reductions to other theories, word equation splitting, and automata-based reasoning. Many automata-based solvers apply constraint propagation, typically using limited forms of automata splitting over concatenation.

Kaluza~\cite{kaluza} and HAMPI~\cite{hampi} handle string constraints over bounded-length variables by translating them into bit-vector constraints. Similarly, G-Strings~\cite{gstrings} and Gecode+s~\cite{gecode} model strings as arrays of integers along with explicit length constraints, relying on constraint propagation in those domains. nfa2sat~\cite{nfa2sat} and Woorpje~\cite{woorpje} encode string problems into propositional logic and solve them via SAT solvers, although Woorpje is limited to bounded-length word equations while nfa2sat supports more general cases.
Sloth~\cite{sloth}, Qzy~\cite{qzy}, and SLOG/Slent~\cite{slog,slent} use automata to represent string constraints and reduce the problem to a model checking task—often solved using tools like IC3~\cite{ic3}. PISA~\cite{pisa}, on the other hand, maps string constraints to monadic second-order logic on finite strings handling \texttt{indexOf} and \texttt{replace} operations.

The Z3Str family—comprising Z3Str/2/3/4~\cite{z3str2, z3str21, z3str3, z3str4}—iteratively decomposes constant strings into substrings and splits variables into subvariables until they become bounded. Z3str2 introduces new search heuristics and detects overlapping variables, while Z3str3 extends this with theory-aware branching to prioritize easier constraints. Z3\textsmaller[1]{STR}3RE~\cite{z3str3re, z3str3re1} further refines the approach by handling regular constraints directly without reducing them to word equations. Building on these ideas, Z3-alpha~\cite{z3alpha1,z3alpha} employs an SMT strategy synthesis technique to select appropriate solving strategies based on the underlying solver (e.g., Z3 or Z3Str4).
Z3~\cite{z3} integrates multiple techniques—including those from Z3-alpha and Z3seq~\cite{Z3seq}, which is based on symbolic automata and derivatives—to capture the theory of sequences. 
While many state-of-the-art string solvers use rewriting as a preprocessing step or to handle specific string functions, cvc5~\cite{cvc5} distinguishes itself by relying on derivation rules as a core strategy, lazily reducing complex string functions to word equations.

Other solvers focus on automata-based techniques and constraint propagation. Early ideas\cite{cp03} applied simplistic regular constraint propagation, while later work~\cite{Min05} introduced forward propagation for web page analysis, supporting complex string functions and transductions. 
STRANGER~\cite{stranger} uses both forward and backward propagation to over-approximate the inputs of string variables at various program points. Norn~\cite{norn,norn1} guarantees termination on the acyclic fragment through automata splitting and length propagation.

Trau~\cite{trau0,trau1, trau2} extends Norn's approach with counterexample-guided abstraction refinement to approximate regular languages using arithmetic formulas and is the first solver to handle the chain-free fragment~\cite{chain19}, supporting features such as string transductions and \texttt{replaceAll}.
CertiStr~\cite{certistr} adopts a forward propagation strategy over regular constraints and concatenation. 
OSTRICH~\cite{popl19} employs refined backward propagation—along with word equation splitting, length abstraction, and letter counting—to propagate regular constraints across string functions (including \texttt{replaceAll} and transductions), thereby guaranteeing completeness for the straight-line fragment.
In practice, OSTRICH is run at SMT-COMP using a portfolio configuration with timeslicing, where three different base solvers are executed sequentially within a fixed time budget: the basic OSTRICH (which combines backward propagation with equation splitting), CEA-OSTRICH~\cite{cea} (which uses cost-enriched automata for improved reasoning about length, \texttt{indexOf}, and \texttt{subString} constraints), and a solver based on algebraic data types~\cite{adt}.
Z3-Noodler~\cite{noodler,z3noodler, z3noodler1} introduces a specialized stabilization algorithm, applying a combination of forward and backward propagation to a selected word equation until the inferred regular languages on both sides stabilize. 
Its techniques have been extended to handle length constraints, disequalities, and string-to-integer conversion. 
Although complete for the chain-free fragment (on supported string functions), Z3-Noodler does not support string transductions and \texttt{replaceAll}—a gap our solver addresses.

Owing to the interest in and the demands of string solving, a standard file
format has been recently agreed and adopted as part of SMT-LIB 2.6~\cite{smtlib2.6}, called 
the theory of Unicode strings (introduced in 2020). This has made it possible to create a track
on string problems at the SMT-COMP~\cite{SMT-COMP}.\footnote{Not all aforementioned solvers
	support SMT-LIB 2.6 format.}


\section{Conclusion and Future Work}
\label{sec:conc}
In this paper, we demonstrated that a simple strategy like 
RCP can effectively solve a wide range of benchmarks in string solving. We first established its completeness on the orderable fragment, which subsumes two of the largest known decidable fragments: \emph{chain-free} and \emph{straight-line} constraints.

We introduced a fair proof search strategy for applying RCP, ensuring that every constraint is eventually propagated and no propagation rule is indefinitely postponed.  
However, the efficiency of proof search may vary depending on the structure of the input constraints. For instance, the Marking Algorithm provides a specialized strategy that is complete for the orderable fragment but does not generalize beyond it.  
Furthermore, our strategies are not without limitations. They are not complete for certain simple equations, and our exploration has not utilized the entire proof system. There are rules that are theoretically not fully understood yet, such as the ``Cut'' rule, which allows splitting the search around arbitrary regular constraints. While this rule can be advantageous in many cases, its integration poses a challenge as it is not clear how to algorithmically choose a constraint to split around.
Developing adaptive heuristics for proof search, tailored to specific classes of string constraints, remains an interesting direction for future work.

Beyond our theoretical contributions, we implemented \solver{} based on our proof system and evaluated its performance on standard SMT string benchmarks. Although \solver{} does not outperform state-of-the-art solvers on the overall SMT-COMP datasets, it remains competitive, solving up to 90\% of the benchmarks—compared to other solvers that solve between 75\% and 99\% of the set. 
Notably, our work significantly improves the performance of base OSTRICH, increasing its benchmark-solving ratio from 75\% to 90\%. When combined with the remainder of the OSTRICH portfolio, the resulting virtual best portfolio achieves up to 97\% of solved benchmarks—second only to Z3-Noodler on the SMT-COMP benchmarks. 
Moreover, \solver{} excels in previously unsolved benchmark categories, namely on PCP instances and on bioinformatics benchmarks, where only OSTRICH achieves comparable results. 
Finally, integrating \solver{} with existing solvers significantly enhances their performance, reducing the number of unsolved benchmarks by 25\% to 74\% on SMT-COMP datasets and by up to 100\% on both PCP and bioinformatics benchmarks.

While our work primarily focuses on string constraints over regular languages, an interesting direction for future research is to explore the applicability of our techniques in the theory of sequences, where the underlying languages are specified by parametric automata. 
As noted in the related work, preliminary investigations using constraint propagation methods in this domain have yielded promising results. 
Further exploration could provide valuable insights into both the theoretical and practical implications of extending our approach to sequence constraints.

s
\section*{Acknowledgments}

This work was supported by the Engineering and Physical Sciences Research Council (EPSRC) under Grant No.~EP/T00021X/1;
the European Research Council (ERC)
under Grant No.~101089343 (LASD);
and the Swedish Research Council
under Grant No.~2021-06327.

\section*{Data Availability Statement}
 The complete implementation of \solver{}, along with the benchmark suites (including the SMT-LIB, PCP, and bioinformatics benchmarks) are available on Zenodo \cite{zenodopower}.



\bibliographystyle{plain}
\bibliography{refs}

\appendix
\section{Proofs for Section~\ref{sec:guarantee}}

RCP assumes that all equational constraints are in a normal form,
i.e.\ they are of the form $x = f(y_1, \ldots, y_k)$, where $f$ uses concatenation only,
or $x = \mathcal{T}(y)$, where $\mathcal{T}$ is a rational function.
We show that the natural translation to the normal form,
in which we replace an equational constraint $t = \mathcal{T}(t')$ with
$x = t, y = t', x = \mathcal{T}(y)$,
preserves the chain-freeness:

\begin{lemma}
	\label{lem:normal_form_chain_free}
	A set of equational constraints $\chi$ is chain-free if and only if the set of its normal forms $\psi$ is chain-free.
\end{lemma}

\begin{proof}
	We show the proof by induction on the number of equations turned into normal form,
	i.e.\ we fix one constraint $\phi_j = `t_{2j-1} = \mathcal T(t_{2j})'$ and replace it with three
	equations
	\begin{eqnarray*}
		x &=& t_{2j-1}\\
		y &=& t_{2j}\\
		x &=& \mathcal{T}(y)
	\end{eqnarray*}
	and show that $G$ has a cycle if and only if $G'$ has a cycle.
	Let $x_1, \ldots, x_n$ be all variables used in both $t_{2j-1}$ and $t_{2j}$.
	For simplicity of proof we will use the same names as in $G$ for the corresponding nodes in $G'$;
	the only new nodes are those corresponding to $x, y$ (two nodes each).
	For simplicity, let them be called $(x, 1)$ for the node in $x = t_{2j-1}$, $(y, 1)$ for $y = t_{2j}$
	and $(x,2), (y,2)$ for the last equation.
	
	Consider an edge $(2j-1,i) \to (j', i')$ (so from node corresponding to a variable in $t_{2j-1}$ to some other variable) in $G$.
	Then in $G'$ there is a path $(2j-1,i) \to (x,2) \to (y, 1) \to (j', i')$, moreover,
	\begin{itemize}
		\item $(x,2)$ has incoming edges only of the form $(2j-1,i) \to (x,2)$;
		\item $(x,2)$ has only one  outgoing edge: $(x,2) \to (y, 1)$ and this is the only incoming edge of $(y, 1)$;
		\item $(y,1)$ has outgoing edges to $(j', i')$ if and only if there is an edge $(2j-1,i) \to (j', i')$ (this is independent of the choice of $i$).
	\end{itemize}
	So, in a sense, this edge was replaced with a path of length $3$ (this is not strict in the sense that $(x,2) \to (y, 1)$ is shared by many such paths).

	Similarly, consider an edge $(2j,i) \to (j', i')$ (so from node corresponding to a variable in $t_{2j}$ to some other variable) in $G$.
	Then in $G'$ there is a path $(2j,i) \to (y,2) \to (x, 1) \to (j', i')$,
	moreover:
	\begin{itemize}
		\item $(y, 2)$ has incoming edges only of the form $(2j,i) \to (y,2)$;
		\item $(y, 2)$ has only one  outgoing edge: $(y,2) \to (x, 1)$ and this is the only incoming edge of $(x, 1)$;
		\item $(x, 1)$ has outgoing edges to $(j', i')$ if and only if there is an edge $(2j,i) \to (j', i')$ (this is indepoendent of the choice of $i$).
	\end{itemize}
	So, as above, this edge was replaced with a path of length $3$.

	
	It is easy to see that if there is a cycle in $G$ then there is one in $G'$ (obtained by replacing appropriate edges by paths; they may in general share node and edges).

In the other direction, if there is a cycle in $G'$ then we can create a cycle in $G$:
any edge that does not use $(x, 1), (x, 2),  (y, 1), (y, 2)$ exists also in $G$.
If the cycle uses, say, $(y,1)$ then it needs to use a path of the form $(2j-1,i) \to (x,2) \to (y, 1) \to (j', i')$ (for some $i, j' , i'$)
and this corresponds to an edge $(2j-1,i) \to (j', i')$ in $G$,
the other case are shown in the same way.
Iterating this process yields a path in $G$.
\end{proof}

\begin{proof}[proof of Lemma~\ref{lemma:fwd}]
	The claim follows from
	\[
	z \in f(L_1, \ldots, L_n) \equiv \exists y_1,\ldots, y_n \; z = f(y_1,\ldots ,y_n) \land \bigwedge_{i=1}^n y_i \in L_i \enspace ,
	\]
	which is the definition of the image of $f$.
%
	%
	%
\end{proof}

\begin{proof}[proof of Lemma~\ref{lemma:bwd}]
	Denote by $\psi$ $\psi'$ the first and second formula.
	
	$\Rightarrow$: Assume that $\psi$ holds.
	Then there is a model that satisfies $\psi$ and there exists a $z$ that satisfies $z \in L_z$ and $z = f(y_1,\ldots ,y_n)$.
	We know that $f$ is backwardable and by assumption $f^{-1}(L_z) = \bigcup^k_{j=1} L_{1,j}\times \dots \times L_{n,j}$,
	hence there exists at least one $j$ such that
	$(y_1, \ldots, y_n) \in L_{1,j}\times \dots \times L_{n,j}$ and $z = f(y_1, \ldots, y_n)$.
	Therefore one  of the disjuncts in $\psi'$ is true and so $\psi'$ is also true.

	$\Leftarrow$: Assume that $\psi'$ holds.
	Then there is a model that satisfies it.
	Since $f^{-1}(L_z) = \bigcup^k_{j=1} L_{1,j}\times \dots \times L_{n,j}$
	and at least one of the disjuncts in $\psi'$ is true,
	for some $j$ we have
	$(y_1, \ldots, y_n) \in L_{1,j}\times \dots \times L_{n,j}$
	and so $f(y_1, \ldots, y_n) \in L_z$, we take it as a witness $z$;
	note that this might be a different witness than in $\psi'$.
	Hence, the same model satisfies $\psi$.
%
%
\end{proof}

\section{Proofs for Section~\ref{sec:lower}}

We first show how the equation and regular constraints
can look like in different sequents of the proof system applied on~\eqref{eq:lower_bound_example} .

\begin{lemma}
	\label{lem:Nielsen_rules_application}
	Consider a sequent in a branch of a proof
	(in a proof system using only rules \texttt{Nielsen}, \texttt{Cut}, \texttt{Intersect}, \texttt{Fwd-Prop}, \texttt{Bwd-Prop} and \texttt{Close}).
	used on example~\eqref{eq:lower_bound_example} such that 
	\texttt{Nielsen} was applied $\ell$ times at the prefix and $r$ times at the suffix. 
	Then the equation is
	$$
	xy x^{r + \ell} xy = yx x^{r + \ell} yx
	$$
	Moreover, the regular constraints are of the form $x \in e_x, y \in e_y$,
	where $L(e_x) \subseteq a^+  \; \land \; L(e_y) \subseteq a^* b a^*$.
\end{lemma}

\begin{proof}
	We show the claim by induction on the number of rules' application.
	It clearly holds for the initial equation for  $r = \ell = 0$.

	If we apply a Cut rule or regular constraint propagation then either we end up with a contradiction (terminating a branch)
	or we make the constraint for $x, y$ stricter, i.e.\ we restrict the regular set,
	so the upper-bounds on them still hold; we also do not modify the equation.
	
	Suppose that we apply the Nielsen's transform to the prefix.
	Taking $x / y$ implies an update in the regular constraint for $x$, i.e.\
	$e_x , e_y$, which is clearly empty, as each string in $L(e_y)$ contains $b$ 
	and no string in $L(e_x)$ contains $b$.

	Taking $x / yx$ implies an update in the regular constraint for $x$,
	i.e.\ $e_x '$ which by~\eqref{eq:ex_over_approximate}
	should satisfy
	$L(e_x') \subseteq \{v \: : \: \exists u \, uv \in L(e_x) \land u  \in L(e_y)\}$,
	which is empty, as each string in $L(e_y)$ contains $b$ 
	and no string in $L(e_x)$ contains $b$.

	Hence the only possible branch uses $y/xy$
	and the equation $xy x^{\ell+r} xy = yx x^{\ell+r} yx$ is replaced (after removing of the two leading $x$'s) with
	$$
	xy x^{r + \ell+1} xy = yx x^{r + \ell+1} yx \enspace ,
	$$
	as claimed.
	Note that the analysis for the application of Nielsen's transform to the suffix is similar,
	with the difference that we consider substitutions $x / xy$  and $y / yx$,
	but the resulting equation is the same.

	For the regular constraint observe that we update the constraint $e_y$,
	and by \eqref{eq:ex_over_approximate}
	it should hold that $L(e_y') \subseteq L(e_x)^{-1}L(e_y)$ or $L(e_y') \subseteq L(e_y)L(e_x)^{-1}$,
	but as $L(e_x) \subseteq a^+, L(e_y) \subseteq a^*ba^*$ this ends in a (perhaps empty) subset of $a^*ba^*$,
	so the claim still holds.
\end{proof}

In the following we represent the regular constraints on $x, y$ using the lengths of strings from $a^*$;
roughly speaking, the set of those lengths are unions of
arithmetic progressions with the same step.
We say that a set $A$ is finite (co-finite) in $B$,
when $A \cap B$ is finite ($B \setminus A$ is finite). 	

\begin{proof}[proof of Lemma~\ref{lem:regular_representation}]
	Consider first $L_x$.
	By known characterizations of unary regular languages, say Chrobak normal form of DFA recognizing a unary regular language,
	we get that $L_x$ is recognized by a DFA which is a finite path followed by a cycle of $m$ states,
	then the accepting states on the finite path correspond to $a^{k_i}$
	and the accepting states on the cycle yield  sets of the form $a^{k_i} \akm$ for appropriate
	$k_i$, and the whole $L_x$ is a union of such languages.
	Now, if there is some $k_i \equiv_m 0$ with $L_i = a^{k_i} \akm$
	then $L_x$ is co-finite in $\akm$ and otherwise it is finite in \akm.

	If $m' = \ell m$ then we can easily compute the representation for $m'$
	from the one for $m$, as $\akm = \bigcup_{k = 0}^{\ell-1} a^{km} \left(a^{m'}\right)^*$.
	Consequently, if we are given several such representations then it is enough
	to take new $m$ as the smallest common multiplier of all such $m$s
	and use the remark above.

	For $L_y$ observe that we cannot simply claim that $L_y = L_y' b L_y''$
	for some regular sets $L_y', L_y''$,
	as there could be some connection between lengths before and after $b$ in $L_y$.
	We can resolve this by a simple argument:
	given an NFA recognizing $L_y$ we find all, say $t$, transitions by $b$ in it
	and make a separate copy of the NFA for each such a transition by $b$,
	in which this transition remains and all other $b$-transitions are removed,
	clearly $L_y =\bigcup_{i = 1}^t L_{y,i}$,
	where $L_{y,i}$ is the language for the automaton with $i$th $b$-transition remaining (and other removed).
	Then $L_{y,i} = L_{y,i}' b L_{y,i}''$ for appropriate regular languages $L_{y,i}', L_{y,i}'' \subseteq a^*$
	and we can represent them as union as in the argument above.
	Using distributivity, we get the required form,
	note that we need to use common $m$ for $L_{y,i}', L_{y,i}''$,
	which is done as before.
	%
	Finally, we need to take one $m$ for $L_x, L_y$, which is done in the same way.
	Other claims follow in the same way as for $L_x$.
\end{proof}

\begin{proof}[proof of Lemma~\ref{lem:infinite_to_cofinite_lemma}]
	The claim for $L_x$ is clear, as by Lemma~\ref{lem:regular_representation}
	it is either finite or co-finite in \akm.
	Similarly, represent $L_y$
	as 
	$$
	L_y = \bigcup_j L_{j,\ell} b L_{j,r} .
	$$
	Then for some $j$ we have that infinitely many $y_i$
	are in $L_{j,\ell} b L_{j,r}$
	and hence infinitely many 
	$a^{k_{i,\ell}m}$ are in $L_{j,\ell}$ and
	infinitely many $a^{k_{i,r}m}$ are in $L_{j,r}$.
	By Lemma~\ref{lem:regular_representation}
	each $L_{j, \ell}, L_{j,r}$ 
	is either finite or co-finite in \akm,
	so both are in fact co-finite in \akm, which ends the proof of the claim.
	
\end{proof}

\begin{proof}[proof of Theorem~\ref{thm:propagation_isnot_enough}]
	We show that if a sequent contains
	regular constraints $x \in e_x, y \in e_y$ (for simplicity we will represents perhaps many constraints as one,
	which can be obtained by intersecting them),
	where the representation of $L(e_x), L(e_y)$ according to Lemma~\ref{lem:regular_representation}
	(for a common $m$) are
	$L(e_x) = \bigcup_i L_{x,i}$ and $L(e_y) = \bigcup_j L_{y,\ell,j}bL_{y,r,j}$ then 
	for some $i, j$ the
	$L_{x,i}, L_{y,\ell,j}, L_{y,r,j}$ are all co-finite in $\akm$.

	In many case the exact index is not needed and we will refer
	to those sets as $L_{x}, L_{y,\ell}, L_{y,r}$.


	We show the claim of the Lemma by induction on the number $m$ of times a rule is applied;
	clearly it holds at the beginning, as for $m = 1$ we have $L_{x} = L_{y,\ell} = L_{y,r} = a^*$.

	For a Cut rule: In a Cut rule we consider an arbitrary regular expression $e$
	(defining a language language $R$)
	and consider separately two branches:
	the constraint on $y$ replaced with $e_y'$ such that
	$L(e_y') = L_y \cap R$ and in the other case: with $L_y \cap \overline R$
	(we can also perform a Cut for $e_x$, in the following we discuss the one for $e_y$ as it is more involved, the analysis in case of $e_x$ can be done in a similar, yet simplified way).
	Consider $R' = R \cap a^*  b a^*$ and $R'' = \overline R \cap a^*  b a^*$,
	clearly $L_y \cap R = L_y \cap R'$ and $L_y \cap \overline R = L_y \cap R''$ and $R' \cup R' = a^* b a^*$.
	Represent $L_y, R', R''$ as in Lemma~\ref{lem:regular_representation} with a common $m$.
	By assumption there are $L_{y,\ell}, L_{y,r}$ co-finite in \akm such that $L_y \supseteq L_{y,\ell} b L_{y,r}$.
	
	Consider the language
	$$
	\{a^{km} b a^{km} \: : \: k \in \mathbb N\} \subseteq \akm b \akm = R' \cup R''
	$$
	At least one of $R', R''$ contains infinitely many element of it, say $R'$ does.
	Then by Claim~\ref{lem:infinite_to_cofinite_lemma}
	there are $L_\ell, L_r$ co-finite in $\akm$ such that $L_\ell b L_r \subseteq R'$.
	Then also $L_\ell \cap L_{y,\ell}$, $L_r \cap L_{y,r}$ are co-finite in \akm
	and clearly
	$$
	(L_{y,\ell}\cap L_\ell) b (L_{y,r} \cap L_r) \subseteq L_y \cap R' \subseteq L_y \cap R \enspace ,
	$$
	as required.

	For a Nielsen's rule:
	by Lemma~\ref{lem:Nielsen_rules_application} only rules $y / xy$ or $y / yx$
	do not lead to immediate contradiction; consider the former,
	the other is symmetric.
	Then the regular constraints $x \in e_x$ and $y \in e_y$
	are updated, into $e_x^i, e_y^i$ (on $i$-th branch).
	Take $m$ as in Lemma~\ref{lem:regular_representation},
	which is common for $L(e_x), L(e_y)$ and all
	$L(e^i_x), L(e^i_y)$.
	By assumption there are $L_x, L_{y,\ell}, L_{y,r}$, all co-finite in \akm,
	such that $L_x \subseteq L(e_x)$ and $L_{y,\ell} b L_{y,r} \subseteq L(e_y)$.
	Consider the triples $(a^{km}, a^{2km},a^{km})$,
	co-finitely many of them are in $L_x \times L_{y,\ell} \times L_{y,r}$
	and so $x_k:=a^{km} \in L_x, y_k:=a^{2km}ba^{km} \in L_{y,\ell}b L_{y,r}$
	for each $k \geq k_0$ for some $k_0$;
	define $y_k'$ such that $y_k = x_ky_k'$, i.e.\
	$y_k'=a^{km}ba^{km}$.
	By~\eqref{eq:no_underapproximation} for each $k \geq k_0$ there is $i$ such that
	$x_k \in L(e_x^i), y'_k  \in L(y_x^i)$,
	so for some $i$ we have that for infinitely many
	$k$ we have $x_k \in L(e_x^i)$ and $y'_k  \in L(e_y^i)$; fix this $i$.
	Then by Claim~\ref{lem:infinite_to_cofinite_lemma}
	$L_x^i$ is co-finite in \akm
	and there are $L_{y,\ell}^i, L_{y,r}^i$ co-finite in \akm
	such that $L_{y,\ell}^i b L_{y,r}^i \subseteq L(e_y^i)$, as required.
	The argument for rule $y \gets yx$ is symmetric.

	For regular constraint propagation:
	formally, the format of the equations required by RCP,
	i.e.\ the normal form,
	forces us to introduce new variable $z$
	and rewrite the equation
	$xyxy = yxyx$ as $z = xyxy$ and $z = yxyx$
	and then $z$ is used as a vehicle to transfer the regular constraints from one side to the other (using forward and backward propagation).
	For the purpose of the proof we will use a common generalization of both (forward and backward) propagation,
	for simplicity we will use only two variables, the generalization to many variables is clear,
	but not needed for the proof:
	\[
	\prftree[r]{
		$\begin{array}{l}
			\text{if } f^{-1}(g (L(e_x), L(e_y))) 
			\ = \bigcup^k_{i=1} (L(e^i_x) \times L(e^i_y))
		\end{array}$
	}{
		\{
		\Gamma,
		x \in e_x, x \in e^i_x, y \in e_y, y \in e^i_y
		\}^k_{i=1}
	}{
		\Gamma, x \in e_x, y \in e_y, g(x,y) = f(x,y)
	}
	\]
	where here $f, g$ are string functions using only concatenation,
	so in particular their image of regular sets is regular and and pre-image of regular sets 
	are recognizable relations.
	It can be seen as applying first the forward and then backward propagation:
	we compute the regular constraint on $g(x,y)$ (forward propagation)
	and then transfer it to $f(x,y)$ (backwards propagation).

	By Lemma~\ref{lem:Nielsen_rules_application} the considered equation is of the form:
	$xy x^{r + \ell}x y = yx x^{r + \ell} yx$ for some $r, \ell$.

	The analysis is symmetric for the sides,
	so let us apply the propagation from the left-hand side
	to the right-hand side.
	Let $L(e_x) = L_x, L(e_y) = L_y$,
	define $f(L_1, L_ 2) = L_2 L_1 L_1^{\ell + r} L_2 L_1$
	and $g(L_1, L_ 2) = L_1 L_2 L_1^{\ell + r} L_1 L_2$;
	that is, the equation can be represented as
	$g(x,y) = f(x,y)$.
	By definition we can branch into new pairs of constraints $\{e_x^i, e_y^i\}_{i\in I}$,
	on $x, y$, when
	$$
	\bigcup_{i \in I} L(e_x^i) \times L(e_y^i) = f^{-1} (\underbrace{L_x L_y L_x^{\ell + r} L_x L_y}_{=g(L_x,L_y)})
	$$

	Consider the representation of $L_x, L_y, \{L_x^i, L_y^i\}_{i\in I}$ according to Lemma~\ref{lem:regular_representation} for a common $m$.
	We want to show that for some $i$ the $L_x^i, L_{y,\ell}^i, L_{y,r}^i$
	are all co-finite in \akm
	and $L_{y,\ell}^ib L_{y,r}^i \subseteq L_y^i$.
	
	By assumption there are $L_{y,\ell}, L_{y,r}$
	such that $L_x, L_{y,\ell}, L_{y,r}$ are co-finite in \akm
	and $L_{y,\ell}b L_{y,r} \subseteq L_y$.
	Then
	\begin{align*}
		g(L_x,L_y)
		&\supseteq
		g(L_x,L_{y,\ell}b L_{y,r})\\
		&=
		\underbrace{L_x L_{y,\ell}}_{\text{co-finite in }\akm}b \underbrace{L_{y,r} (L_x)^{\ell + r} L_x L_{y,\ell}}_{\text{co-finite in }\akm}b \underbrace{L_{y,r}}_{\text{co-finite in }\akm}
		\enspace .
	\end{align*}
	Now consider
	\begin{align*}
		x_i
		&=
		a^{im}\\
		y_i
		&=
		a^{im}ba^{im}\\
		f(x_i,y_i)
		&=
		y_ix_ix_i^{\ell+r}y_ix_i\\
		&=
		\underbrace{a^{im}}_{\in \akm}b
		\underbrace{a^{(\ell+r+3)im}}_{\in \akm}b
		\underbrace{a^{2im}}_{\in \akm} .
	\end{align*}
	Hence, except for finitely many elements,
	$f(x_i,y_i) \in g(L_x,L_y)$
	and so for co-finitely many $i$s we have
	
	\begin{align*}
		(x_i,y_i) &\in f^{-1}(g(L_x,L_y)) \\
		&=
		\bigcup_{j \in I} L(e_x^j) \times L(e_y^j) 
	\end{align*}
	In particular, for some $j$ infinitely many $x_i$ belong to
	$L(e_x^j)$, which is then co-finite in \akm{} by Claim~\ref{lem:infinite_to_cofinite_lemma},
	and infinitely many $y_i$ belong to $L(e_y^j)$,
	so by Claim~\ref{lem:infinite_to_cofinite_lemma} there are $L_{y,\ell}^j, L_{y,r}^j$ co-finite in \akm
	such that $L_{y,\ell}^j b L_{y,r}^j \subseteq L(e_y^j)$, as required.
\end{proof}


\section{Proof Search Strategies} \label{sec:proof-strategies}

We begin with some observations.
Every equational constraint $x = f(\bar{y})$ defines two propagation rules:$(x = f(\bar{y}), \leftarrow)$ and $(x = f(\bar{y}), \rightarrow)$.
The set of propagation rules derived from a string constraint \( \psi \) consists of all such propagation rules for each equational constraint occurring in \( \psi \).

Note, none of the proof rules adds equational constraints to a string constraint.
For the sake of presentation and simplicity, we assume that each string variable is associated with exactly one regular constraint.  This assumption is justified by the fact that multiple constraints can be combined through \texttt{Intersect}, ensuring a unified representation without loss of generality.
If no such constraint is explicitly given, we assume \( \Sigma^* \) as a default constraint.

We extend the definition of a \emph{proof tree} in Section \ref{sec:proof_sys} to an \emph{annotated proof tree}, where each node is annotated with propagation rules and associated clocks. We call the clocks associated with a leaf and propagation rule, \emph{propagation clock} and the clocks only associated with a leaf \emph{branch clock}.
\begin{definition}
	Let $T = (V, E, P, B)$ be an \emph{annotated proof tree}, where $(V,E)$ forms the underlying proof tree. The annotation $P$  is defined as follows:
	$$ P: L \to 2^{\mathcal{P} \times \mathbb{N}}, $$
	where $L$ is the set of leaf nodes in $T$, \( \mathcal{P} \) is the set of all propagation rules induced by occurring equational constraints in $T$ and \( \mathbb{N} \) tracks clock values.
	The annotation $B$  is defined as follows:
	$$ B: L \to \mathbb{N}.$$
\end{definition}

For a given string constraint  $\psi = \bigwedge_{i = 1}^{n} \phi_i$, the initial proof tree has a single root node $\Gamma = \phi_1, \ldots, \phi_n$.
The annotation assigns $P(\Gamma) = \{ ((\phi, \leftarrow), 0), ((\phi, \rightarrow), 0) \mid \phi \in \Gamma \}$.
That is, the clock for each propagation rule is $0$.
Furthermore, the annotation assigns $B(\Gamma) = 0$. That is, the clock for the initial branch is set to 0.

We say an annotated proof tree is \emph{closed} if the underlying proof tree is closed.
Next, we formally define what it means to apply a proof rule in the context of constructing a proof tree. More precisely, we define how a proof tree is expanded through the application of a proof rule.
We start with the rule \texttt{Close} which closes branches of the proof tree.

\begin{definition} \label{def:close-expansion}
	Let \( T = (V, E, P, B) \) be an \emph{annotated proof tree}, where \( P, B \) are the annotation functions.
	Given a leaf node \( \Gamma \in V \), suppose there exists a regular constraint \( \phi := x \in \emptyset \) such that \( \phi \in \Gamma \).
	
	Then \( \Gamma \) satisfies the condition for the \textnormal{\texttt{Close}} rule:
	$$
	\frac{}{\Gamma} \quad \textnormal{\texttt{[Close]}}
	$$
	We define the \emph{expansion} $\mathcal{E}$ of \( T \) with respect to \( \Gamma \) using \textnormal{\texttt{Close}} as the function:
	$$
	\mathcal{E}(T, \Gamma, \textnormal{\texttt{Close}}) = T',
	$$	
	
	The updated proof tree \( T' = (V', E', P', B') \) is constructed as follows:
	\begin{itemize}
		\item The new set of nodes is given by \( V' = V \cup \{\bot\} \), where \( \bot \) represents a fresh dead-end node.
		\item The set of edges is updated to \( E' = E \cup \{ (\Gamma, \bot) \} \).
		\item The annotation function \( P', B' \) are updated as follows:
		\begin{itemize}
			\item The new dead-end node \( \bot \) is also marked as closed:
			$$
			P'(\bot) = \emptyset \text{ and } B'(\bot) = 0.
			$$
			\item For all other leaf nodes \( \Gamma' \neq \Gamma \), the annotation remains unchanged:
			$$
			P'(\Gamma') = P(\Gamma') \text{ and } B'(\Gamma') = B(\Gamma').$$
		\end{itemize}
	\end{itemize}
\end{definition}

\begin{definition} \label{def:expansion}
	Let \( T = (V, E, P, B) \) be an \emph{annotated proof tree}, where \( P, B \) are the annotation functions.  
	Given a leaf node \( \Gamma \in V \) containing the equation \( \phi_f := x = f(\bar{y}) \), suppose there exists a backward propagation rule \( (\phi_f, \leftarrow) \) with clock value \( c \), such that \( ((\phi_f, \leftarrow), c) \in P(\Gamma) \). Additionally, assume that \( \Gamma \) includes the regular constraint \(\phi_e := x \in e \).

	If \( \phi_f, \phi_e \) under sequents \( S_1, \dots, S_n \) form a valid proof rule:
	$$
	\frac{S_1 \quad S_2 \quad \dots \quad S_n}{x = f(\bar{y}), x \in e} \quad \textnormal{\texttt{[Bwd-Prop]}}
	$$
	Then, we define the \emph{expansion} $\mathcal{E}$ of \( T \) with respect to $\Gamma$ and $(\phi_f, \leftarrow)$ as the function:
	$$
	\mathcal{E}(T,\Gamma, (\phi_f, \leftarrow)) = T',
	$$	
	
	First for each $i \in \{1,\ldots,n\}$ we can define $S_i' := S_i \cup \Gamma$.	
	Then we construct a proof tree \( T' = (V', E', P', B') \) where:
	\begin{itemize}
		\item The new set of nodes is given by \( V' = V \cup \{ S_1', \dots, S_n' \} \).
		\item The set of edges is updated as \( E' = E \cup \{ (\Gamma, S_i') \mid i = 1, \dots, n \} \).
		\item The annotation function \( P' \) is updated as follows:
		\begin{itemize}
			\item Each new node \( S_i' \) inherits the remaining annotations from \( \Gamma \) and resets the clock of \( (\phi_f, \leftarrow) \), increments all other propagation clocks, and initializes a new branch clock:
			\[
			\begin{aligned}
				P'(S_i') =	&\left\{ ((\phi, d), c+1) \mid ((\phi, d), c) \in P(\Gamma), d \in \{\leftarrow, \rightarrow\}, \phi \neq \phi_f \right\} \\
				&\quad \cup \{ ((\phi_f, \leftarrow), 0) \} \\
				&\quad \cup \{ ((\phi_f, \rightarrow), c+1) \mid ((\phi_f, \rightarrow), c) \in P(\Gamma) \}. \\
				B'(S_i') = & 0
			\end{aligned}
			\]

			\item The propagation clocks of all other leaves $\Gamma' \notin \{ S_1', \dots, S_n'\}$ stay the same and the branch clocks are incremented by 1:
			
			$
			P'(\Gamma') = P(\Gamma')\\
			B'(\Gamma') = B(\Gamma') +1
			$
		\end{itemize}
	\end{itemize}
\end{definition}

The expansion of a proof tree using \texttt{Fwd-Prop} can be defined analogously to the backward propagation case. 
We define the expansion $\mathcal{E}$ of \( T \) with respect to \( (\phi, \rightarrow) \) as the function:
$$
\mathcal{E}(T, \Gamma, (\phi, \rightarrow)) = T',
$$
which produces a new annotated proof tree \( T' = (V', E', P', B') \) from \( T \) by applying an analogous expansion process for forward propagation.
Next, we define a proof search as a sequence of proof tree expansions, formally described as follows.

\begin{definition}
	A \emph{proof search} on a set of string constraints \( \Gamma \) is a sequence of annotated proof trees
	$$
	T_0, T_1, T_2, \dots
	$$  
	where the initial tree is given by \( T_0 = (\{\Gamma\}, \emptyset, P_0, B_0) \), consisting of a single node \( \Gamma \) with the annotation functions  
\begin{align*}
	P_0(\Gamma) &= \{ ((\phi, \leftarrow), 0), ((\phi, \rightarrow), 0) \mid \phi \in \Gamma \}, \\
	B_0(\Gamma) &= 0.
\end{align*}

	For each subsequent tree \( T_{i+1} \), there exists either:
	\begin{itemize}
		\item A propagation rule \( (\phi, d) \) in the annotation function of some leaf node \( \Gamma' \), where \( d \in \{\leftarrow, \rightarrow\} \), such that  
		$$
		T_{i+1} = \mathcal{E}(T_i, \Gamma', (\phi, d)),
		$$  
		or
		\item A leaf node \( \Gamma' \) containing a regular constraint of the form \( x \in \emptyset \), in which case the \textnormal{\texttt{Close}} rule is applied:
		$$
		T_{i+1} = \mathcal{E}(T_i, \Gamma', \textnormal{\texttt{Close}}).
		$$
	\end{itemize}
\end{definition}

Finally, we can give the fairness definition.

\begin{definition}
	A proof search \( T_0, T_1, T_2, \dots \) is \emph{fair} if it satisfies the following two conditions:
	\begin{enumerate}
		\item For each propagation rule \( (\phi, d) \) with \( d \in \{\leftarrow, \rightarrow\} \), there exists a constant \( c_j > 0 \) such that whenever
		\[
		((\phi, d), c) \in P_i(\Gamma_i) \quad \text{for some leaf node } \Gamma_i \in V_i \text{ in some } T_i,
		\]
		it holds that
		\[
		c \leq c_j.
		\]
		\item For each leaf node \(\Gamma \in V_i\) that is not closed (i.e., extendable), there exists a constant \( c_\Gamma > 0 \) such that, if \(\Gamma\) remains as a leaf node in some subsequent proof tree \( T_j \) (with \( j \geq i \)), then its branch clock value satisfies
		\[
		B_j(\Gamma) \leq c_\Gamma.
		\]
	\end{enumerate}
\end{definition}

This condition ensures that for a given bound, no propagation rule is indefinitely postponed during the proof search. 
For example, consider a string constraint $\psi = \bigwedge_{i=1}^{n} \phi_i,$
where each $\phi_i$ is an equational constraint. In this case, a natural choice for an upper bound on the propagation clock is $2n$. This ensures that each propagation rule is applied at least once along any single branch of the proof tree.

In contrast, the upper bounds for the branch clocks must be chosen dynamically, since the width of the proof tree can increase as new branches are created by the application of proof rules. For a given proof tree $T_i$, a simple bound is $|L_i|$, where $L_i$ denotes the set of open (i.e., not closed) leaf nodes in $T_i$. This bound guarantees that every branch is expanded at least once before any branch is revisited for further expansion.

\end{document}